\documentclass[11pt]{article}

\usepackage{algpseudocode} %
\usepackage{amsmath,amsthm,amssymb} %
\usepackage{authblk} %
\usepackage[english]{babel} %
\usepackage{bbm} %
\usepackage{booktabs} %
\usepackage{color} %
\usepackage{derivative} %
\usepackage{dsfont} %
\usepackage{enumitem} %
\usepackage{float} %
\usepackage{flushend} %
\usepackage[T1]{fontenc} %
\usepackage{framed} %
\usepackage[top=1in,bottom=1.1in,left=1.5in,right=1.5in]{geometry} %
\usepackage{graphicx} %
\usepackage[backref,colorlinks,citecolor=blue,bookmarks=true,breaklinks=true]{hyperref} %
\usepackage[capitalize]{cleveref} 
\usepackage{hyphenat} %
\usepackage{ifdraft} %
\usepackage[utf8]{inputenc} %
\usepackage{interval} %
\usepackage{mathrsfs} %
\usepackage{mathtools} %
\usepackage{mdframed} %
\usepackage{microtype} %
\usepackage{soul} %
\usepackage{times} 
\usepackage{tikz} %
\usepackage{suffix} 
\usepackage{xparse} %
\usepackage{xspace} %

\setul{1ex}{.5pt} %

\setlist[enumerate]{nolistsep,itemsep=3pt,topsep=3pt,leftmargin=*} %
\setlist[itemize]{nolistsep,itemsep=3pt,topsep=3pt,leftmargin=2em} %

\intervalconfig{soft open fences}

\mdfdefinestyle{figstyle}{ %
  linecolor=black!8, %
  backgroundcolor=black!8, %
  innertopmargin=.8em, %
  innerleftmargin=.8em, %
  innerrightmargin=.8em, %
  innerbottommargin=.8em %
}

\newfloat{algorithm}{hbt}{lop} %
\floatname{algorithm}{Algorithm} %
\crefname{algorithm}{Algorithm}{Algorithms} %

\newfloat{metaalgorithm}{hbt}{lop} %
\floatname{metaalgorithm}{Meta-Algorithm} %
\crefname{metaalgorithm}{Meta-Algorithm}{Meta-Algorithms} %

\newenvironment{algspec}{ %
  \begin{mdframed}[style=figstyle]}{ %
  \end{mdframed}} %

\theoremstyle{plain} 
\newtheorem{theorem}{Theorem}[section] %
\newtheorem{lemma}{Lemma}[section] %
\newtheorem{proposition}{Proposition}[section] %
\newtheorem{corollary}{Corollary}[section] %
\theoremstyle{definition} %
\newtheorem{definition}{Definition}[section] %
\theoremstyle{remark} %

\newcommand{\microspace}{\mspace{.5mu}} %
\newcommand{\ket}[1]{\ensuremath{\lvert\microspace#1%
    \microspace\rangle}} %
\newcommand{\bigket}[1]{\bigl\lvert\microspace#1%
  \microspace\bigr\rangle} %
\newcommand{\bra}[1]{\ensuremath{\langle\microspace#1%
    \microspace\rvert}} %
\newcommand{\bigbra}[1]{\bigl\langle\microspace#1%
  \microspace\bigr\rvert} %

\newcommand{\ignore}[1]{} %
\newcommand{\ip}[2]{\ensuremath{\left\langle#1,#2\right\rangle}} %
\newcommand{\norm}[1]{\ensuremath{\left\lVert #1 \right\rVert}} %
\newcommand{\abs}[1]{\ensuremath{\left\lvert #1 \right\rvert}} %
\newcommand{\defeq}{\stackrel{\smash{\text{\tiny\rm def}}}{=}} %
\newcommand{\real}{\mathbb{R}} %
\newcommand{\at}[2][]{#1|_{#2}} %

\newcommand{\dom}[1]{\mathrm{dom}\,{#1}} %
\newcommand{\interior}[1]{\mathrm{int} (#1)} %
\newcommand{\intdom}[1]{\interior{\dom{#1}}} %
\newcommand{\epi}[1]{\mathrm{epi}\,#1} %
\newcommand{\cl}[1]{\mathrm{cl}\left(#1\right)} %

\newcommand{\class}[1]{\textup{#1}\xspace} %

\newcommand{\MIP}{\class{MIP}} %
\WithSuffix\newcommand\MIP*{\ensuremath{\class{MIP}^*}} %

\newcommand{\setft}[1]{\mathrm{#1}} %
\newcommand{\Density}{\setft{D}} %
\newcommand{\Pos}{\setft{Pos}} %
\newcommand{\Herm}{\setft{Herm}} %

\newcommand{\labelstyle}[1]{\mathrm{#1}} %
\newcommand{\Shannon}{\labelstyle{H}} %
\newcommand{\lint}{\labelstyle{int}} %
\newcommand{\lbd}{\labelstyle{bd}} %
\newcommand{\li}{\labelstyle{i}} %
\newcommand{\lb}{\labelstyle{b}} %

\newcommand{\abl}{\labelstyle{Average}} %

\newcommand{\pf}{\mathscr{P}} 
\newcommand{\lf}{\mathscr{L}} 
\newcommand{\clpf}{\cl{\pf}} 
\newcommand{\aux}{\mathcal{A}} 

\newcommand\restr[2]{{
  \left.\kern-\nulldelimiterspace
  #1 
  \vphantom{\big|} 
  \right|_{#2} 
  }}

\DeclareMathOperator{\tr}{tr} %
\DeclareMathOperator*{\argmin}{argmin} %
\DeclareMathOperator*{\argmax}{argmax} %
\DeclareMathOperator{\supp}{supp} %
\DeclareMathOperator{\eig}{eig} %

\def\e{e} 
\def\I{\mathds{1}} 
\def\L{\mathcal{L}} 
\def\X{\mathcal{X}} 

\begin{document}


\title{\Large\bf Classical and Quantum Iterative Optimization Algorithms Based on
  Matrix Legendre-Bregman Projections}

\renewcommand*{\Affilfont}{\small\itshape} 
\author{Zhengfeng Ji} %
\affil{ %
  Department of Computer Science and Technology,\protect\\
  Tsinghua University, Beijing, China
}

\date{\today}

\maketitle

\begin{abstract}
  We consider Legendre-Bregman projections defined on the Hermitian matrix space and design iterative optimization algorithms based on them.
  A general duality theorem is established for Bregman divergences on Hermitian matrices and it plays a crucial role in proving the convergence of the iterative algorithms.
  We study both exact and approximate Bregman projection algorithms.
  In the particular case of Kullback-Leibler divergence, our approximate iterative algorithm gives rise to the non-commutative versions of both the generalized iterative scaling (GIS) algorithm for maximum entropy inference and the AdaBoost algorithm in machine learning as special cases.
  As the Legendre-Bregman projections are simple matrix functions on Hermitian matrices, quantum algorithmic techniques are applicable to achieve potential speedups in each iteration of the algorithm.
  We discuss several quantum algorithmic design techniques applicable in our setting, including the smooth function evaluation technique, two-phase quantum minimum finding, and NISQ Gibbs state preparation.
\end{abstract}

\section{Introduction}\label{sec:intro}


Bregman divergence is a quantity introduced in~\cite{Bre67} to solve convex
optimization problems under linear constraints.
It is a distance-like measure defined for a convex function and generalizes the
Euclidean distance.
An important example of Bregman divergence beyond the squared Euclidean distance
is the Kullback-Leibler information divergence.
Bregman divergence and several related concepts play a crucial role in many
areas of study, including optimization theory~\cite{Byr07}, statistical learning
theory~\cite{LPP97,CSS02,D-PD-PL97,BMDG05}, and information theory~\cite{CS04}.


Given a convex optimization problem over a convex set $C = \bigcap_{j} C_{j}$
specified as the intersection of potentially simpler convex sets $C_{j}$,
Bregman's method iteratively projects the initial point to different $C_{j}$'s
using Bregman's divergence as the measure of distance.
This method, referred to as Bregman's projection algorithm~\cite{Bre67}, is one
of the most influential iterative algorithms for solving convex optimization
problems under linear constraints~\cite{Byr07}.
It generalizes the orthogonal projection algorithms in the Euclidean space and
the Euclidean distance is replaced by the Bregman divergence for certain
underlying convex functions.
Many important convex optimization algorithms are special cases of Bregman's
iterative projection algorithm by choosing different divergence measures.
When Bregman's divergence is chosen to be the Kullback-Leibler
divergence, Bregman's projection is also known as information
projection~\cite{CS04} and several iterative algorithms, including the
generalized iterative scaling algorithm~\cite{DR72} (also known as the SMART
algorithm~\cite{GBH70}), are information projection algorithms~\cite{CS04}.


Bregman's algorithm iteratively computes the Bregman's projection on convex
sets, the exact computation of which is a complicated task and usually does not
have explicit formulas even in the simple case of linear constraints.
Several different ideas have been proposed to compute the Bregman's projection
approximately~\cite{Cen81,D-PI86,CD-PE+90,CD-PI91,CH02,CSS02,D-PD-PL97,D-PD-PL02}.
Of great importance to this work is the auxiliary function
method~\cite{D-PD-PL97,CSS02,D-PD-PL02} that designs an auxiliary function to
bound the progress of the iterative update procedure measured by Bregman
divergence.
Important learning algorithms including the improved iterative scaling
algorithm~\cite{D-PD-PL97} and the AdaBoost
algorithm~\cite{FS97,CSS02,D-PD-PL02} can be analyzed using the auxiliary
function method and shown to converge to the correct optimizer.


This paper considers \emph{non-commutative} analogs of Bregman's
projection algorithms where Bregman divergence is defined for Hermitian
matrices.
For a real convex function $f$, and two matrices $X$, $Y$, we consider the
Bregman divergence $D_{f}(X, Y) = \tr \bigl(f(X) - f(Y) - f'(Y)(X-Y))$.
Even though generalizing the Bregman's projection and related concepts to the
matrix case is straightforward, analyzing their behavior is a challenging task.
This is primarily due to two reasons.
On the one hand, matrices are generally \emph{non-commutative}, so inequalities
used in the analysis become much harder to establish.
On the other hand, the proof of convergence in the classical case relied on the
continuity of Bregman divergence, while Bregman divergence for matrices
is usually \emph{discontinuous}.
A natural question we ask here is: ``\textit{Are there natural generalizations of
  the iterative Bregman projection algorithms that converge correctly given the
  difficulty posed in the non-commutative case?}''
We answer the question in the positive by establishing two main results in the
paper.


First, we prove a \emph{general duality theorem} in the non-commutative case.
The duality theorem is concerned with the optimization of $D_{f}(X, Y)$ in two
situations.
On one hand, we minimize $D_{f}(X, Y_{0})$ over $X$ in a linear family defined
as $\lf= \bigl\{ X \mid \ip{F_{j}}{X} = \ip{F_{j}}{X_{0}}, \text{ for
} j = 1, 2, \ldots, k \bigr\}$.
That is, we compute the projection of $Y_{0}$ to the linear family under the
Bregman divergence $D_{f}$.
On the other hand, we consider minimization $D_{f}(X_{0}, Y)$ over $Y$ in a
family of Hermitian matrices called the Legendre-Bregman projection family
$\pf = \bigl\{ Y \mid Y = \L_{f}(Y_{0}, \lambda \cdot F) \bigr\}$ where
$\L_{f}(Y, \Lambda) = (f^{*})'(f'(Y) + \Lambda)$.
The duality theorem states that optimizers in the above two situations coincide
under simple and easy-to-verify conditions on $f$.
An important special case of the duality theorem when $f(x) = x \ln x - x$ is
the well-known result that the linear family and the closure of the exponential
family intersect at a point that maximizes the entropy function under linear
constraints~\cite{CS04,AAKS21}.
This special case is well-known as the Jaynes' maximum entropy
principle~\cite{Jay57}, which states that the maximum entropy state satisfying
linear constraints is the unique intersection of the linear family defined by
the constraints and the closure of the exponential family.
The duality theorem is the key to proving the convergence of our exact (and
approximate) matrix Bregman projection algorithms.


Second, we prove a matrix inequality that is essential for analyzing the
approximate information projection algorithms.
In the classical case, Jensen's inequality suffices to establish the properties
we require for the auxiliary function.
In the non-commutative case, the corresponding inequality is much harder to
establish.
In fact, the inequality is not always true for all convex functions.
Fortunately for us, we can prove the inequality for Kullback-Leibler divergence
by employing a strengthened version of Golden-Thompson inequality recently
established by Carlen and Lieb~\cite{CL19}.
This new inequality, together with the duality theorem, guarantees the
convergence of approximate information projection algorithms.

Important examples of this algorithm include the \emph{matrix AdaBoost
  algorithm} and the \emph{matrix generalized iterative scaling (GIS) algorithm}
as special cases.
The matrix AdaBoost algorithm has a physical interpretation.
It is an iterative algorithm that minimizes the partition function of the linear
family of Hamiltonians.
The matrix GIS algorithm is an algorithm for maximum entropy inference and can
be applied in understanding many-body quantum systems~\cite{CJZZ12,Alh22}.
The convergence of these algorithms follows as they are special cases of the
information projection algorithm.


Several potential quantum speedups are identified for the iterative algorithms
proposed in the paper.
As the computation in each iterative step of our algorithms boils down to the
computation of a matrix $Y' = \L_{f}(Y, \Lambda)$ and then an update of
parameters based on average values of the matrix $Y'$ with respect to given
operators $F_{j}$.
We can represent $Y$ as a quantum state and then employ quantum algorithmic
techniques such as singular value transformation~\cite{GSLW19,Gil19} and smooth
function evaluation~\cite{AGGW20} to compute $Y'$, the updated version of $Y$.

Thanks to the flexibility of the auxiliary function framework, we can update the
parameters either in parallel or sequentially, and the analysis of them can be
treated uniformly.
The \emph{sequential} update variant has the advantage of utilizing the fast
quantum OR lemma for searching a violation with low sample complexity of the
states representing the matrix $Y' = \L_{f}(Y, \Lambda)$.
As this state preparation is usually the most expensive part of the computation,
the saving in samples could lead to substantial speedups.

When the state preparation for $Y' = \L_{f}(Y, \Lambda)$ can be done on a
near-term quantum device, our algorithms can also be performed on the same
near-term device because of the intrinsic iterative structure and simple update
rules.
This gives rise to near-term applications for certain convex optimization
problems.


\subsection{Techniques}

To prove the general duality theorem, we followed the approach
in~\cite{D-PD-PL02} with important changes to avoid the problems caused by the
discontinuity of $D_{f}(X, Y)$ and to simplify the assumptions on $f$.
Assumption A.3 in~\cite{D-PD-PL02} requires that $D_{f}(X, Y)$ is continuous
with respect to $X$ and $Y$, a condition holds in the classical case but fails
in the quantum case where $X, Y$ are matrices.
In fact, an explicit example is given in~\cite[Example 7.29]{BB97} showing that
$\lim_{t \to \infty} D(Y, Y_{t}) \ne 0$ for a sequence ${(Y_{t})}_{t}$
converging to $Y$.
We get around the difficulty by extending the domain of the Bregman divergence
carefully and establishing the required properties directly for the extended
versions without using continuity of $D_{f}(X, Y)$.
We also identify \emph{one simple condition} that supersedes the five
assumptions A.1--A.5 of~\cite{D-PD-PL02}.
The simple condition requires that the underlying convex function $f$ has an
open conjugate domain $\dom{f^{*}}$.
The key consequence for a convex function having an open conjugate domain is
that we can show $D_{f}(X, Y)$ is \emph{coercive} with respect to $Y$.
That is, the set $\{ Y \mid D_{f}(X, Y) \le c \}$ is bounded for any constant
$c$, a condition key to show compactness and convergence later on.
Such a condition is equivalent to the function being a so-called
Bregman-Legendre function~\cite{BB97} for real functions.
However, it is also shown in the same paper that when the function is extended
to Hermitian matrices, the function is not Bregman-Legendre because of the
discontinuity issue.
We use techniques from matrix perturbation theory to prove that such a condition
suffices to guarantee the validity of the duality theorem.
There are many interesting convex functions that satisfy the condition, and we
have listed a few important ones in \cref{tab:examples}~\cite{BB97}.

The continuity of $\L_{f}(Y, \Lambda)$ is still essential for our proof
of the convergence analysis, and it is one of the main technical parts of our
proof.
In the end, the problem is roughly a matrix perturbation problem where we
have a block matrix $
\begin{pmatrix}
  A_{0} & B\\
  B^{\dagger} & A_{1}
\end{pmatrix}
$ where $A_{0}$ has small eigenvalues and $A_{1}$ has large eigenvalues and
$\norm{B} \le 1$.
We need to show that the perturbation of the spectrum by $B$ becomes arbitrarily
smaller when the eigenvalues of $A_{0}$ go to infinity.
For this, we make use of two results from perturbation theory that take care of
the perturbation of the eigenvalues (a result by Mathias~\cite{Mat98}) and the
perturbation of the eigenprojections (Davis-Kahan $\sin(\Theta)$
theorem~\cite{DK70}).

For the second result about information projection algorithms, the main
difficulty we encounter is to show appropriate matrix inequalities so that we
can bound the improvement measured by the change in the Bregman divergence using
the auxiliary function.
We are able to show the inequality for an important case when the divergence is
Kullback-Leibler information divergence.
The resulting iterative update algorithm has a very similar flavor to the matrix
multiplicative weight update method (MMWU)~\cite{Kal07,AHK12}.
It is an adaptive version where the update step size is not fixed as in MMWU,
but depends on the violation at the current step.
MMWU is proved to be a powerful framework in designing both classical and
quantum algorithms for semi-definite programming (SDP) problems.
The QIS algorithm are applicable as a replacement for MMWU in some cases and
provide possible speedups thanks to its adaptive nature of the QIS algorithm.
The technical inequality for the MMWU analysis is the Golden-Thompson
inequality.
In contrast, the Golden-Thompson inequality does not seem to be powerful enough
any more for the information projection algorithms, even in combination with
Jensen's operator and trace inequalities.
Luckily for us, an improved Golden-Thompson inequality established recently by
Carlen and Lieb~\cite{CL19} fits our analysis perfectly.

Jaynes' maximum entropy principle was a crucial fact that recent
studies~\cite{AAKS20,AAKS21,HKT21} on the Hamiltonian learning problems heavily
rely on.
The QIS algorithm serves as a candidate algorithm that solves a related problem
that we call the Hamiltonian inference problem.
Both problem tries to learn information about the Hamiltonian.
In the Hamiltonian learning problem, the algorithm is provided copies of the
Gibbs state of the \emph{true Hamiltonian}.
While in the Hamiltonian inference problem, the algorithm does not have access
to the true Gibbs state, but only has information about the local information of
it and is allowed to make \emph{adaptive} queries to the Gibbs state of
candidate Hamiltonians in the linear family of the Hamiltonians.
For local Hamiltonians, it is likely that the QIS algorithm can solve the
Hamiltonian inference problem with not only polynomial sample complexity but
with polynomial time complexity as well.
We leave the analysis to future work.

\section{Preliminary}\label{sec:prelim}


We will need the following concepts from convex analysis.
In convex analysis, functions are defined on all of $\real^{m}$ and take values
from $\real \cup \{\pm \infty\}$.
The \emph{(effective) domain} of a function
$\phi : \real^{m} \rightarrow \real \cup \{\pm \infty\}$ is the set
\begin{equation*}
  \dom{\phi} = \{ x \in \real^{m} \mid \phi(x) < +\infty \}.
\end{equation*}
A function $\phi$ is \emph{convex} if its domain $\dom{\phi}$ is a convex set
and it satisfies
\begin{equation*}
  \phi(t x + (1-t) y) \le t \phi(x) + (1-t) \phi(y)
\end{equation*}
for all $t \in \interval[open]{0}{1}$ and $x, y \in \dom{\phi}$.
The function $\phi$ is \emph{strictly convex} if the above inequality is strict.

For two real vectors $x, y \in \real^{m}$, we define $x \cdot y$ as
$\sum_{i=1}^{m} x_{i} y_{i}$.
For matrices $A, B$, define $\ip{A}{B} = \tr(A^{\dagger} B)$.


The \emph{Fenchel conjugate} $\phi^{*}$ of a convex function
$\phi$ is defined as
\begin{equation}
  \label{eq:conjugate}
  \phi^{*}(y) = \sup\, \{ x \cdot y - \phi(x) \mid x \in \real^{m} \}.
\end{equation}


Let $f$ be a smooth real function and $A(x)$ a matrix whose entries are
functions of $x$.
Then
\begin{equation}
  \label{eq:dif-trace}
  \odv*{\tr f(A(x))}{x} = \ip{f'(A(x))}{\odv*{A(x)}{x}}.
\end{equation}


We will need several results from matrix perturbation theory in our proofs.
Let $A$ be an Hermitian matrix of size $m+n$ by $m+n$ and has a block form $A =
\begin{pmatrix}
  A_{0} & 0\\
  0 & A_{1}
\end{pmatrix}
$ where $A_{0}$ and $A_{1}$ are $m$ by $m$ and $n$ by $n$ Hermitian matrices
respectively.
Let $\tilde{A} =
\begin{pmatrix}
  A_{0} & B\\
  B^{\dag} & A_{1}
\end{pmatrix}
$ be a perturbation of $A$.
We have the following two eigenvalue and eigenvector perturbation bounds.

\begin{proposition}[Eigenvalue Perturbation Bound~\cite{Mat98}]%
  \label{pro:eigenvalue-perturb}
  Let $A$ and $\tilde{A}$ be Hermitian matrices given above.
  Let $\lambda_{k}$ and $\tilde{\lambda}_{k}$ be the $k$-th largest eigenvalue
  of $A$ and $\tilde{A}$ respectively.
  Suppose the eigenvalues of $A_{0}$ and $A_{1}$ are separated in the sense that
  $\lambda_{\min}(A_{1}) - \lambda_{\max}(A_{0}) \ge \eta > 0$.
  Then for all $k = 1, 2, \ldots, m+n$,
  \begin{equation*}
    \abs{\lambda_{k} - \tilde{\lambda}_{k}} \le \frac{{\norm{B}}^{2}}{\eta}.
  \end{equation*}
\end{proposition}

\begin{proposition}[Davis-Kahan $\sin(\Theta)$ Theorem~\cite{DK70}]%
  \label{pro:eigenvector-perturb}
  Let $A$ and $\tilde{A}$ be Hermitian matrices in $\Herm{\X}$ given above.
  Let $V_{0}$ and $V_{1}$ be two isometries mapping into $\X$ whose ranges are
  two orthogonal eigenspaces of $\tilde{A}$ and let $V =
  \begin{pmatrix}
    V_{0} & V_{1}
  \end{pmatrix}
  $.
  Write $\tilde{A} =
  V
  \begin{pmatrix}
    \Lambda_{0} & 0\\
    0 & \Lambda_{1}
  \end{pmatrix}
  V^{\dag} $.
  Define $C_{0}, S_{0}, C_{1}, S_{1}$ to be the submatrices of $V$ as $V =
  \begin{pmatrix}
    C_{0} & - S_{1} \\
    S_{0} & \phantom{-}C_{1}
  \end{pmatrix}
  $.
  Suppose the eigenvalues of $A_{0}$ and $\Lambda_{1}$ are separated in the
  sense that $\lambda_{\min}(\Lambda_{1}) - \lambda_{\max}(A_{0}) \ge \eta > 0$.
  Then
  \begin{equation*}
    \norm{S_{0}} \le \frac{\norm{B}}{\eta}.
  \end{equation*}
\end{proposition}

\subsection{Inequalities}

\begin{lemma}[Jensen's operator inequality]\label{lem:jensen-operator}
  For operator convex function $f$, operators $A_i$ satisfying $\sum_i A_i^\dag
  A_i = \I$ and $X_i \in \Herm(\X)$, the following inequality holds
  \begin{equation*}
    f \Bigl( \sum_i A_i^\dag X_i A_i \Bigr) \preceq
    \sum_i A_i^\dag f(X_i) A_i.
  \end{equation*}
\end{lemma}

\begin{lemma}[Jensen's trace inequality]\label{lem:jensen-trace}
  For convex function $f$, operators $A_i$ satisfying $\sum_i A_i^\dag
  A_i = \I$ and $X_i \in \Herm(\X)$, the following inequality holds
  \begin{equation*}
    \tr f \Bigl( \sum_i A_i^\dag X_i A_i \Bigr) \le
    \tr \sum_i A_i^\dag f(X_i) A_i.
  \end{equation*}
\end{lemma}

\begin{lemma}[Golden-Thompson inequality]
  For Hermitian matrices $A$ and $B$, it holds that
  \begin{equation*}
    \tr \e^{A+B} \le \tr \bigl(\e^{A} \e^{B}\bigr).
  \end{equation*}
\end{lemma}

\begin{lemma}[Carlen-Lieb inequality]\label{lem:carlen-lieb}
  Suppose $H$ is an Hermitian matrix, $Y \succ 0$.
  Then the following inequality holds
  \begin{equation}
    \label{eq:carlen-lieb}
    \tr \exp(\ln(Y) + H) \le \exp \bigl( \inf \bigl\{
    \lambda_{\max}(H-\ln Q) \, : \, Q \succ 0, \tr(YQ)=1
    \bigr\}\bigr).
  \end{equation}
\end{lemma}

We remark that Carlen-Lieb is a strengthening of Golden-Thompson as pointed out
in~\cite{CL19}.
This can be seen by choosing $Y = e^{A}$, $H = B$, and
$Q = e^{B} / \tr(e^{A}e^{B})$ in Carlen-Lieb.

\subsection{Bregman Divergence}\label{sec:bregman}


Bregman divergence is an important quantity in convex analysis and information
theory.
In this section, we recall its definition, and discuss the definition of
Legendre functions for which the Bregman divergence behaves nicely.

\begin{definition}\label{def:breg-div}
  Let $\phi$ be a convex function such that $\phi$ is differentiable on
  $\intdom{\phi}$.
  The \emph{Bregman divergence}
  $D_{\phi} : \dom{\phi} \times \intdom{\phi} \rightarrow
  \interval[open right]{0}{+\infty}$
  is defined as
  \begin{equation}
    \label{eq:breg-div}
    D_{\phi}(x, y) = \phi(x) - \phi(y) - \nabla \phi (y) \cdot (x-y).
  \end{equation}
\end{definition}


The Bregman divergence is sometimes also known as the Bregman distance because
$D_{\phi}(x, y)$ is a natural measure of the distance between $x, y$ even though
it is not necessarily a distance in the sense of metric topology (e.g.,
$D_{\phi}(x, y)$ is in general not symmetric with respect to $x$, $y$).
For example, when $\phi(x) = {\norm{x}}^{2}$, $D_{\phi}(x, y)$ recovers the
squared Euclidean distance ${\norm{x-y}}^{2}$.
It also holds that $D_{\phi}(x, y) \ge 0$ and equality holds if and only if
$x = y$ for strictly convex $\phi$.


We now define the Bregman projection of a point to a convex set.
The definition is natural in the geometric picture that Bregman divergence
generalizes the squared Euclidean distance.

\begin{definition}[Bregman Projection]
  Let $C$ be a closed convex set in $\real^{m}$ such that $C \cap \dom{\phi}$ is
  not empty.
  Then the Bregman projection of $y$ to $C$ is defined as
  \begin{equation*}
    y^{*} = \argmin_{x \in C \cap \dom(\phi)} D_{\phi}(x, y).
  \end{equation*}
\end{definition}


Next, we will define an important family of functions called Legendre functions.
For this, we need the following technical definitions about convex functions.
A convex function is \emph{proper} if never takes the value $-\infty$ and takes
a finite value for at least one $x$.
Most of the time, we work with proper convex functions.
A convex function $\phi$ is \emph{closed} if its epigraph
$\epi{\phi} = \{(x, t)\in \real^{m+1} \mid x \in \dom{\phi}, \phi(x) \le t\}$ is
closed.
A proper convex function $\phi$ is \emph{essentially smooth} if it is everywhere
differentiable on the interior of the domain $\intdom{\phi}$ and if
$\norm{\nabla \phi(x_{j})}$ diverges for every sequence $(x_{j})$ in
$\intdom{\phi}$ converging to a point on the boundary of $\dom{\phi}$.

\begin{definition}\label{def:leg-type}
  A function $\phi$ is \emph{Legendre} if it is a proper closed convex function
  that is essentially smooth and strictly convex on the interior of its
  domain.
\end{definition}

A fundamental fact about Legendre convex functions is recorded below.

\begin{proposition}[Theorem 26.5 of~\cite{Roc70}]\label{pro:leg-dif}
  If $\phi$ is Legendre then
  $\nabla \phi : \intdom{\phi} \rightarrow \intdom{\phi^{*}}$ is a bijection,
  continuous in both directions, and $\nabla \phi^{*} = {(\nabla \phi)}^{-1}$.
\end{proposition}

The Fenchel conjugate of a closed convex function is again a closed convex
function.
For proper closed convex function $\phi$, ${(\phi^{*})}^{*} = \phi$.
A convex function is Legendre if and only if its conjugate is.


As an example, consider the extended real function
$\phi_{\Shannon}: \real \rightarrow \real \cup \{\pm \infty\}$
\begin{equation*}
  \phi_{\Shannon}(x) =
  \begin{cases}
    x\ln x - x &\text{for}\ x > 0,\\
    0 &\text{for}\ x=0,\\
    +\infty &\text{otherwise}.
  \end{cases}
\end{equation*}
It is easy to verify that $\phi_{\Shannon}$ is Legendre and has domain
$\dom{\phi_{\Shannon}} = \interval[open right]{0}{+\infty}$.
The conjugate $\phi_{\Shannon}^{*}$ of $\phi_{\Shannon}$ is
$\phi_{\Shannon}^{*}(y) = \e^{y}$ with domain
$\dom{\phi_{\Shannon}^{*}} = \real$.
The derivative $\nabla \phi_{\Shannon}(x) = \ln(x)$ is a bijection between
$\interval[open]{0}{+\infty}$ and $\real$ and is the inverse function of
$\nabla \phi_{\Shannon}^{*}$.
More examples of Legendre functions, their conjugates and effective domains are
given in \cref{tab:examples}.


\begin{table}[htb!]
  \centering
  \begin{tabular}{ccccc}
    \toprule
    $f$ & $\dom{f}$ & $f^{*}$ & $\dom{f^{*}}$ & Remarks\\
    \midrule
    $x^2/2$ & $\real$ & $y^{2}/2$ & $\real$ & Euclidean\\
    $-\sqrt{1-x^{2}}$ & $\interval{-1}{1}$
                    & $\sqrt{1+y^{2}}$ & $\real$ & Hellinger\\
    $x\ln{x} - x$ & $\interval[open right]{0}{\infty}$
                    & $\e^{y}$ & $\real$ & Boltzmann/Shannon\\
    $x\ln{x} + (1-x) \ln(1-x)$ & $\interval{0}{1}$
                    & $\ln(1+\e^{y})$ & $\real$ & Fermi/Dirac\\
    $-\ln x$ & $\interval[open]{0}{\infty}$
                    & $1-\ln(-y)$ & $\interval[open]{-\infty}{0}$ & Burg\\
    \bottomrule
  \end{tabular}
  \caption{Examples of Legendre convex functions}\label{tab:examples}
\end{table}


The importance of Legendre functions is highlighted in~\cite{BB97} where it is
shown that when $\phi$ is Legendre, the Bregman projection defined above exists,
is unique, and belongs to the interior $\intdom{\phi}$.
The fact that for Legendre functions, the Bregman projection is in
$\intdom{\phi}$ is crucial for iterative Bregman projection algorithms
considered in the literature and in this paper because we will project $y^{*}$
again and $D_{\phi}(x,y^{*})$ is only defined for $y^{*} \in \intdom{\phi}$.
This was previously guaranteed by the additional requirement called ``zone
consistency'' and the use of Legendre functions is a great simplification.

\begin{lemma}\label{lem:conjugate-form}
  Let $\phi$ be a Legendre convex function. Then for $y \in \intdom{\phi}$,
  \begin{equation*}
    \phi^{*}(y) = y \cdot \nabla \phi^{*}(y) - \phi( \nabla \phi^{*}(y)).
  \end{equation*}
\end{lemma}

\begin{proof}
  For fixed $y \in \intdom{\phi}$, the function $x \mapsto x \cdot y - \phi(x)$
  is concave and has zero gradient at $x^{*} = \nabla \phi^{*} (y)$ and,
  therefore, achieves the maximum at $x^{*}$.
  The lemma follows by evaluating the function $x \mapsto x \cdot y - \phi(x)$
  at $x^{*}$.
\end{proof}


We will also extensively use Legendre-Bregman conjugate and Legendre-Bregman
projections from~\cite{D-PD-PL02}.

\begin{definition}\label{def:leg-breg-conj-and-proj}
  For a convex function $\phi$ of Legendre type defined on $\real^{m}$, the
  \emph{Legendre-Bregman conjugate} of $\phi$ is
  $\ell_{\phi} : \intdom{\phi} \times \real^{m} \rightarrow
  \real \cup \{\pm \infty\}$
  \begin{equation}
    \label{eq:leg-breg-conj}
    \ell_{\phi} (y, \lambda) = \sup_{x \in \dom{\phi}} \bigl(
    \lambda \cdot x - D_{\phi}(x, y) \bigr).
  \end{equation}
  The \emph{Legendre-Bregman projection}
  $\L_{\phi} : \intdom{\phi} \times \real^{m} \rightarrow \dom{\phi}$ is
  \begin{equation}
    \label{eq:leg-breg-proj}
    \L_{\phi} (y, \lambda) = \argmax_{x \in \dom{\phi}} \bigl(
    \lambda \cdot x - D_{\phi}(x, y) \bigr).
  \end{equation}
\end{definition}

By the definition of $\ell_{\phi}(y, \lambda)$ as the conjugate of function
$D_{\phi}(\;\cdot\;, y)$, it is convex with respect to $\lambda$ for all fixed
$y$, a property useful in later analysis.

The following proposition is crucial for this work and, as $\nabla \phi$ and
$\nabla \phi^{*}$ are inverse to each other, the proposition gave
$\L_{\phi}(y, \lambda)$ the meaning of updating of the parameter by first mapping
$y$ to the parameter space and pulling it back after the update using
$\nabla \phi^{*}$.

\begin{proposition}[Proposition 2.6
  of~\cite{D-PD-PL02}]\label{pro:leg-breg-proj-exp}
  Let $\phi$ be a function of Legendre type.
  Then for $y \in \intdom{\phi}$ and
  $\lambda \in \intdom{\phi^{*}} - \nabla \phi(y)$,
  the Legendre-Bregman projection is given explicitly by
  \begin{equation}
    \label{eq:leg-breg-proj-exp}
    \L_{\phi}(y, \lambda) = (\nabla \phi^{*}) (\nabla \phi(y) + \lambda).
  \end{equation}
  Moreover, it can be written as a Bregman projection
  \begin{equation}
    \label{eq:leg-breg-proj-proj}
    \L_{\phi}(y, \lambda) = \argmin_{x \,\in\, \dom{\phi} \,\cap H} D_{\phi}(x, y)
  \end{equation}
  for the hyperplane
  $H = \{x \in \real^{m} \mid \lambda \cdot (x - \L_{\phi}(y, \lambda)) = 0 \}$.
\end{proposition}

An important corollary of the proposition and \cref{pro:leg-dif} is the
following.
\begin{corollary}\label{cor:additive}
  $\L_{\phi}(y, \lambda)$ is additive in $\lambda$.
  Namely, for
  $\lambda_{1}, \lambda_{1} + \lambda_{2} \in \intdom{\phi^{*}} - \nabla \phi(y)$,
  \begin{equation*}
    \L_{\phi}(\L_{\phi}(y, \lambda_{1}), \lambda_{2}) =
    \L_{\phi}(y, \lambda_{1} + \lambda_{2}).
  \end{equation*}
\end{corollary}

\begin{lemma}\label{lem:boundary-effect}
  Suppose $f$ is a Legendre real convex function with domain
  $\Delta \subseteq \real$ and the domain $\dom{f^{*}}$ is open.
  Let $a$ be a real number in $\cl{\Delta} \setminus \Delta$ and $(y_{t})$ be a
  sequence in $\Delta_{\lint}$ that converges to $a$.
  Let $(x_{t})$ be a sequence in $\Delta$ that converges to $x \in \Delta$.
  Then
  \begin{equation*}
    \lim_{t \to \infty}D_{f}(x_{t}, y_{t}) = +\infty.
  \end{equation*}
\end{lemma}

\begin{proof}
  As $f$ is a real convex function, it is possible to perform a case study for
  the domain $\Delta$.
  We will prove the claim for $\Delta = \interval[open]{a}{\infty}$ and other
  cases can be dealt with similarly.
  Fix any $y \in \interval[open]{a}{x}$, then eventually it will hold that
  $y_{t} < y$ for sufficiently large $t$.
  Thus,
  \begin{equation*}
    f'(y_{t}) \le \frac{f(y_{t}) - f(y)}{y_{t} - y}.
  \end{equation*}
  Hence, we have
  \begin{equation*}
    \begin{split}
      D_{f}(x_{t}, y_{t}) & = f(x_{t}) - f(y_{t}) + f'(y_{t})(y_{t} - x_{t})\\
      & \ge f(x_{t}) - f(y_{t}) + \frac{f(y_{t})-f(y)}{y_{t}-y} (y_{t} - x_{t})\\
      & = f(x_{t}) - f(y) \frac{y_{t}-x_{t}}{y_{t}-y} + f(y_{t})\frac{y-x_{t}}{y_{t}-y},
    \end{split}
  \end{equation*}
  where the first two terms converge and the third term goes to $+\infty$ when
  $t$ goes to $\infty$.
\end{proof}

\section{Bregman Divergence for Hermitian Operators}\label{sec:breg-mat}

Suppose $\X$ is a finite dimensional Hilbert space and $f$ is an extended real
convex function.
In this section, we will always use $\Delta$ to denote the domain $\dom{f}$ of
$f$, the interval on which $f$ takes finite values.
Then $f$ extends to all Hermitian operators in $\Herm_{\Delta}(\X)$ as
\begin{equation*}
  f(X) = \sum_{k} f(\lambda_{k}) \Pi_{k}
\end{equation*}
where $X = \sum_{k} \lambda_{k} \Pi_{k}$ is the spectral decomposition of $X$.
In this paper, we focus on convex functions of the form $\phi = \tr \circ\, f$.
Denote the interior and boundary of $\Delta$ as $\Delta_{\lint}$ and
$\Delta_{\lbd} = \Delta \setminus \Delta_{\lint}$ respectively.
It is easy to see that the domain of $\phi$ is
$\dom{\phi} = \Herm_{\Delta}(\X)$, and the interior of the domain
$\interior{\dom{\phi}} = \Herm_{\Delta_{\lint}}(\X)$.

In this case, the Bregman divergence becomes
\begin{equation}
  \label{eq:breg-div-mat}
  D_{\phi}(X, Y) = \tr(f(X) - f(Y) - f'(Y)(X-Y)),
\end{equation}
defined for $X\in \Herm_{\Delta}(\X)$ and $Y \in \Herm_{\Delta_{\lint}}(\X)$.
The Legendre-Bregman projection can be written explicitly as
\begin{equation}
  \label{eq:leg-breg-proj-mat}
  \L_{\phi}(Y, \Lambda) = (f^{*})' \bigl( f'(Y) + \Lambda \bigr)
\end{equation}
for $Y \in \Herm_{\Delta_{\lint}}(\X)$ and
$\Lambda \in \Herm_{\Delta^{*}_{\lint}}(\X) - f'(Y)$.
Slightly abusing the notation, we use $D_{f}$, $\ell_{f}$ and $\L_{f}$ to denote
$D_{\tr \circ\, f}$, $\ell_{\tr \circ\, f}$, and $\L_{\tr \circ\, f}$
respectively.

\subsection{Domain Extension}\label{sec:extension}

In the previous discussion, $D_{f}(X, Y)$, $\ell_{f}(Y, \Lambda)$, and
$\L_{f}(Y, \Lambda)$, are only defined for $Y \in \Herm_{\Delta_{\lint}}(\X)$.
For later discussions, it is necessary to extend the definition of them to all
operators $Y \in \Herm_{\Delta}(\X)$, allowing the functions to take infinite
values sometimes.


We first introduce several notations used in defining the extensions.
For an operator $A \in \Herm_{\Delta}(\X)$, define $\X_{\lint(A)}$ and
$\X_{\lbd(A)}$ as the spans of eigenspaces of $A$ corresponding to eigenvalues
in $\Delta_{\lint}$ and $\Delta_{\lbd}$ respectively.
Then the Hilbert space $\X$ has decomposition
$\X = \X_{\lint(A)} \oplus \X_{\lbd(A)}$ and the operator $A$ has decomposition
$A = A_{\lint} \oplus A_{\lbd}$ for $A_{\lint} \in \Herm(\X_{\lint(A)})$ and
$A_{\lbd} \in \Herm(\X_{\lbd(A)})$ by grouping its eigenspaces depending on
whether the corresponding eigenvalues are in $\Delta_{\lint}$ or in
$\Delta_{\lbd}$ respectively.
For operator $B \in \Herm_{\Delta}(\X)$, we write $B_{\lint(A)}$ as the
restriction of $B$ to $\X_{\lint(A)}$.
For operators $A, B \in \Herm_{\Delta}(\X)$, define $\X_{\lbd(A \land B)}$ as
the span of common eigenvectors of $A, B$ with the same eigenvalues in
$\Delta_{\lbd}$.
Define $\X_{\lint(A \lor B)}$ as the orthonormal complement of
$\X_{\lbd(A \land B)}$ in $\X$.
For $A, B \in \Herm_{\Delta}(\X)$ and all $t \in \interval[open]{0}{1}$, the
restriction of $(1-t) A + t B$ to $\X_{\lint(A \lor B)}$ has all eigenvalues
contained in $\Delta_{\lint}$.
The support $\supp(A)$ of an operator $A$ is defined to be $\X_{\lint(A)}$.
For a convex set $C \subseteq \Herm_{\Delta}(\X)$, it follows from the convexity
of $C$ that there is an operator $A \in C$ whose support contains the support of
all other operators in $C$.
This is defined to be the support $\supp(C)$ of the convex set $C$.

\begin{definition}
  Let $f$ be a real convex function with domain $\Delta$.
  A pair of operators $X, Y \in \Herm_{\Delta}(\X)$ is said to be
  \emph{admissible} (with respect to $f$), written as $X \triangleright Y$, if
  for all eigenvalues of $Y$ in $\Delta \setminus \Delta_{\lint}$, the
  corresponding eigenspace of $Y$ is contained in the eigenspace of $X$ of the
  same eigenvalue.
  Equivalently, $X \triangleright Y$ if and only if
  $X = X_{\lint(Y)} \oplus Y_{\lbd}$.
\end{definition}


The definition of $D_{f}(X, Y)$ is extended to all
$X, Y \in \dom{\phi} = \Herm_{\Delta}(\X)$ as follows.

\begin{definition}
  Let $f$ be real convex function of Legendre type whose domain is $\Delta$.
  The (extended) Bregman divergence $D_{f}(X, Y)$ for all
  $X, Y \in \Herm_{\Delta}(\X)$ is defined as
  \begin{equation}
    \label{eq:breg-div-ext}
    D_{f}(X, Y) =
    \begin{cases}
      D_{f}(X_{\lint(Y)}, Y_{\lint})\ & \text{if}\ X \triangleright Y,\\
      +\infty\ & \text{otherwise}.
    \end{cases}
  \end{equation}
\end{definition}

This extension follows the convention in information theory where for
$f(x) = x \ln x - x$, $0 \cdot f'(0) = 0 \cdot \ln 0$ is defined to be $0$ so
that relative entropy $D(p \Vert q)$ is defined as long as the support of $p$ is
contained in that of $q$.
The convention makes sense for general real convex functions of Legendre type as
it is not hard to verify that
$\lim\limits_{x \to a} \bigl( (x-a)\, f'(x) \bigr) = 0$ for $x$ converging in
$\Delta_{\lint}$ to $a \in \Delta \setminus \Delta_{\lint}$.

Similarly, we can extend the definition of the Legendre-Bregman projection
$\L_{f}(Y, \Lambda)$ to $Y \in \Herm_{\Delta}(\X)$.

\begin{definition}
  The (extended) Legendre-Bregman projection $\L_{f}(Y, \Lambda)$ is
  \begin{equation}
    \label{eq:leg-breg-proj-ext}
    \L_{f}(Y, \Lambda) = \L(Y_{\lint}, \Lambda_{\lint(Y)}) \oplus Y_{\lbd},
  \end{equation}
  which is defined for all $Y \in \Herm_{\Delta}(\X)$ and $\Lambda$ such that
  \begin{equation}
    \label{eq:Lambda-admissible}
    \Lambda_{\lint(Y)} \in \Herm_{\Delta^{*}_{\lint}}(\X) - f'(Y_{\lint}).
  \end{equation}
  For convenience, we say that $\Lambda$ is \emph{admissible} (with respect to
  $Y$) for $\L_{f}$ if \cref{eq:Lambda-admissible} holds.
\end{definition}

The definition can be justified as follows.
\begin{align}
  \L_{f}(Y, \Lambda)
  & = \argmax_{X \in \Herm_{\Delta}(\X)} \bigl( \ip{\Lambda}{X} -
    D_{f}(X, Y) \bigr)\nonumber\\
  & = \argmax_{\substack{X \in \Herm_{\Delta}(\X)\\ X \triangleright Y }}
  \bigl( \ip{\Lambda}{X} - D_{f}(X, Y) \bigr)\tag{By \cref{eq:breg-div-ext}}\\
  & = \argmax_{\substack{X \in \Herm_{\Delta}(\X)\\
  X = X_{\lint(Y)} \oplus Y_{\lbd}}} \bigl(
  \ip{\Lambda_{\lint(Y)}}{X_{\lint(Y)}} -
  D_{f}(X_{\lint(Y)}, Y_{\lint}) \bigr)\nonumber\\
  & = \L_{f}(Y_{\lint}, \Lambda_{\lint(Y)}) \oplus Y_{\lbd}.
\end{align}
Hence, the extended definition of $\L_{f}$ is consistent with the requirement
that it is the maximizer of $\ip{\Lambda}{X} - D_{f}(X, Y)$.
This also implies that $\supp(\L_{f}(Y, \Lambda)) = \supp(Y)$ by
\cref{pro:leg-dif}.

\subsection{Basic Properties}\label{sec:properties}


\begin{lemma}[Corollaries 3.2 and 3.3 of~\cite{Lew96}]\label{lem:trf-leg}
  Convex function $\phi = \tr \circ f$ is Legendre if and only if $f$ is
  Legendre.
\end{lemma}

\begin{proof}
  This follows from the fact that $\phi = \tr \circ f$ is the composition of
  $\hat{\phi}(x) = \sum_{i}f(x_{i})$ and the spectrum map $\eig$ that returns
  the ordered tuple of eigenvalues of Hermitian operators.
  The proposition follows from Corollaries 3.2 and 3.3 of~\cite{Lew96} as
  $\hat{\phi}$ as the sum of Legendre functions is Legendre.
\end{proof}

%

\begin{lemma}\label{lem:coercive}
  Suppose $f$ is a real convex of Legendre type and the domain of the conjugate
  $\dom{f^{*}}$ is open.
  Then $D_{f}(X, \;\cdot\;)$ is coercive for all $X \in \Herm_{\Delta}(\X)$.
\end{lemma}

\begin{proof}
  Theorem 5.8 of~\cite{BB97} states that for \emph{real} convex functions, the
  condition that $\dom{f^{*}}$ is open is equivalent to BL0 and BL1 defined
  there, and is therefore also equivalent to $f$ being coercive for all
  $x\in \dom{f}$ by the discussion in Remark 5.3 in that paper.
  This shows that $D_{f}(x, \;\cdot\;)$ is coercive for all $x\in \Delta$.
  It remains to show that the coercive property remains true when lifting $f$ to
  $\tr \circ f$.
  In other words, we need to prove that the set
  \begin{equation*}
    \{ Y \mid D_{f}(X, Y) \le c \}
  \end{equation*}
  is bounded for all $c$.

  It suffices to consider the case where $\supp(Y) = \X$ as the general case
  reduces to it by the definition of extended $D_{f}$.
  Let $Y = \sum_{j} y_{j} \ket{\psi_{j}}\bra{\psi_{j}}$ be the spectrum
  decomposition of $Y$ where $\{\ket{\psi_{j}}\}$ is an orthonormal basis of
  $\X$.
  Define
  \begin{equation*}
    \tilde{x}_{j} = \bra{\psi_{j}} X \ket{\psi_{j}},
  \end{equation*}
  the diagonal elements of $X$ in the basis $\ket{\psi_{j}}$ and
  \begin{equation*}
    \tilde{X} = \sum_{j} \tilde{x}_{j} \ket{\psi_{j}} \bra{\psi_{j}}.
  \end{equation*}
  It is obvious that $\tilde{x}_{j} \in \Delta$ for all $j$.

  By the definition of $D_{f}(X, Y)$, we have
  \begin{equation*}
    \begin{split}
      D_{f}(X, Y) & = \tr \bigl( f(X) - f(Y) - f'(Y) (X - Y)\bigr)\\
      & = \tr f(X) - \sum_{j} f(y_{j}) - \sum_{j} f'(y_{j})
      (\tilde{x}_{j} - y_{j})\\
      & = \tr f(X) - \tr f (\tilde{X}) + \sum_{j} D_{f}(\tilde{x}_{j}, y_{j}).
    \end{split}
  \end{equation*}
  Hence, $D_{f}(X, Y) \le c$ implies that for all $j$
  \begin{equation*}
    D_{f}(\tilde{x}_{j}, y_{j}) \le c + \tr f(\tilde{X}) - \tr f(X),
  \end{equation*}
  which further implies that $y_{j}$ is bounded for all $j$ as
  $D_{f}(\tilde{x}_{j}, \;\cdot\;)$ is coercive.
\end{proof}

\begin{lemma}\label{lem:proj-conti}
  For all Legendre $f$, the Legendre-Bregman projection $\L_{f}(Y, \Lambda)$ is
  continuous on its domain.
\end{lemma}

\begin{proof}
  Using \cref{eq:leg-breg-proj-mat}, we have
  $\L_{f}(Y, \Lambda) = (f^{*})' \bigl( f'(Y) + \Lambda \bigr)$.
  As $f$ is Legendre, both $(f^{*})'$ and $f'$ are continuous in the interior of
  the domain $\Delta_{\lint}$ by \cref{pro:leg-dif}.
  Let $(Y^{(n)}, \Lambda^{(n)})$ be a sequence that converges to $(Y, \Lambda)$.
  If $Y$ is in the interior $\Herm_{\Delta_{\lint}}(\X)$, the claim follows from
  the continuity of $f'$ and $(f^{*})'$.
  Otherwise, suppose $Y = Y_{\lint} \oplus Y_{\lbd}$.
  Without loss of generality, assume $Y^{(t)} \in \Herm_{\Delta_{\lint}}(\X)$ and write
  \begin{equation*}
    Y^{(t)} = Y^{(t)}_{\li} \oplus Y^{(t)}_{\lb},
  \end{equation*}
  such that $\lim_{t \to \infty} Y^{(t)}_{\li} = Y_{\lint}$ and
  $\lim_{t \to \infty} Y^{(t)}_{\lb} = Y_{\lbd}$.
  In the basis of the decomposition $Y^{(t)} = Y^{(t)}_{\li} \oplus Y^{(t)}_{\lb}$, write
  $\Lambda^{(t)} =
  \begin{pmatrix}
    \Lambda^{(t)}_{00} & \Lambda^{(t)}_{01}\\
    \Lambda^{(t)}_{10} & \Lambda^{(t)}_{11}
  \end{pmatrix}
  $.
  Hence,
  \begin{equation*}
    \begin{split}
      f' \bigl( Y^{(t)} \bigr) + \Lambda^{(t)}
      & =
      \begin{pmatrix}
        f' \bigl( Y^{(t)}_{\li} \bigr) + \Lambda^{(t)}_{00} & \Lambda^{(t)}_{01}\\
        \Lambda^{(t)}_{10} & f' \bigl( Y^{(t)}_{\lb} \bigr) + \Lambda^{(t)}_{11}
      \end{pmatrix}\\
      & \defeq
      \begin{pmatrix}
        A_{0} & B\\
        B^{\dag} & A_{1}
      \end{pmatrix},
    \end{split}
  \end{equation*}
  where we have omitted the dependence of $A_{0}$, $A_{1}$, and $B$ on $t$ for
  simplicity.
  Let $V$ be the unitary that diagonalizes the matrix
  \begin{equation}
    \label{eq:continuity-L-decomp}
    \begin{pmatrix}
      A_{0} & B\\
      B^{\dag} & A_{1}
    \end{pmatrix}
    = V
    \begin{pmatrix}
      \Gamma_{0} & 0 \\
      0 & \Gamma_{1}
    \end{pmatrix}
    V^{\dag}
  \end{equation}
  so that $\Gamma_{0}$ is a diagonal matrix that has all the small
  eigenvalues of $
  \begin{pmatrix}
    A_{0} & B\\
    B^{\dag} & A_{1}
  \end{pmatrix}
  $ on its diagonal.
  Write $V =
  \begin{pmatrix}
    C_{0} & -S_{1}\\
    S_{0} & C_{1}
  \end{pmatrix}
  $ and expand the block matrix multiplication we have
  \begin{equation}
    \label{eq:continuity-L-AB}
    \begin{split}
      A_{0} & = C_{0} \Gamma_{0} C_{0}^{\dag} + S_{1} \Gamma_{1} S_{1}^{\dag},\\
      B^{\dag} & = S_{0} \Gamma_{0} C_{0}^{\dag} - C_{1} \Gamma_{1} S_{1}^{\dag}.
    \end{split}
  \end{equation}
  The unitarity of $V$ implies
  \begin{equation}
    \label{eq:continuity-L-unit}
    \begin{split}
      C_{0}^{\dag}C_{0} + S_{0}^{\dag}S_{0} & = \I,\\
      C_{1}^{\dag}C_{1} + S_{1}^{\dag}S_{1} & = \I,\\
      S_{1}^{\dag}C_{0} - C_{1}^{\dag}S_{0} & = 0.
    \end{split}
  \end{equation}
  Using the above \cref{eq:continuity-L-AB,eq:continuity-L-unit},
  \begin{equation}
    \label{eq:continuity-L-AB2}
    \begin{split}
      & S_{1}^{\dag}A_{0} - C_{1}^{\dag} B^{\dag}\\
      =\; & \bigl( S_{1}^{\dag}C_{0} - C_{1}^{\dag}S_{0} \bigr)
      \Gamma_{0} C_{0}^{\dag} + \bigl(S_{1}^{\dag} S_{1} +
      C_{1}^{\dag} C_{1} \bigr) \Gamma_{1} S_{1}^{\dag}\\
      =\; & \Gamma_{1} S_{1}^{\dag}.
    \end{split}
  \end{equation}

  By the essential smoothness of the function $f$ and
  \cref{pro:eigenvalue-perturb}, $\Gamma_{1}$ has eigenvalues that all go to
  $\pm \infty$ when $t$ goes to $\infty$.
  And by \cref{pro:eigenvector-perturb}, the norm $\norm{S_{0}}$ (and
  $\norm{S_{1}}$) go to $0$ as they are bounded by $\norm{B}/\eta$ where $\eta$
  is the gap between the eigenvalues of $A_{0}$ and $\Gamma_{1}$ ($\Gamma_{0}$
  and $A_{1}$, respectively), a quantity that goes to $+\infty$ when $t$ goes
  to $\infty$.

  Using \cref{eq:continuity-L-AB,eq:continuity-L-AB2}, we expand $A_{0}$ as
  \begin{equation*}
    A_{0} = C_{0} \Gamma_{0} C_{0}^{\dag} + S_{1} \Gamma_{1} S_{1}^{\dag}
    = C_{0} \Gamma_{0} C_{0}^{\dag} + S_{1} S_{1}^{\dag} A_{0}
    - S_{1} C_{1}^{\dag} B^{\dag},
  \end{equation*}
  and hence
  \begin{equation*}
    \lim_{t \to \infty} A_{0}  = \lim_{t \to \infty} C_{0} \Gamma_{0} C_{0}^{\dag},
  \end{equation*}
  because $A_{0}$ and $B^{\dagger}$ are finite and $\norm{S_{1}}$ goes to $0$ as
  $t$ goes to $\infty$.
  Define $L_{0} = (f^{*})'(\Gamma_{0})$ and $L_{1} = (f^{*})'(\Gamma_{1})$.
  It is easy to check that
  \begin{equation*}
    \begin{split}
      & \lim_{t \to \infty} C_{0} L_{0} C_{0}^{\dag}\\
      =\; & \lim_{t \to \infty} (f^{*})'( C_{0} \Gamma_{0} C_{0}^{\dag})\\
      =\; & (f^{*})' \bigl( f'(Y_{\lint}) + \Lambda_{\lint(Y)} \bigr)\\
      =\; & \L_{f}(Y_{\lint}, \Lambda_{\lint(Y)}),
    \end{split}
  \end{equation*}
  and similarly
  \begin{equation*}
    \lim_{t \to \infty} C_{1} L_{1} C_{1}^{\dag}
    = \lim_{t \to \infty }(f^{*})'(f'(Y_{\lb}^{(t)})) = Y_{\lbd}.
  \end{equation*}

  Finally, it follows from the definition of $\L_{f}$ that
  \begin{equation*}
    \begin{split}
      \L_{f}(Y^{(t)}, \Lambda^{(t)}) & = V
      \begin{pmatrix}
        (f^{*})'(\Gamma_{0}) & 0\\
        0 & (f^{*})'(\Gamma_{1})
      \end{pmatrix}
      V^{\dag}\\
      & =
      \begin{pmatrix}
        C_{0} & -S_{1}\\
        S_{0} & C_{1}
      \end{pmatrix}
      \begin{pmatrix}
        L_{0} & 0\\
        0 & L_{1}
      \end{pmatrix}
      \begin{pmatrix}
        C_{0}^{\dag} & S_{0}^{\dag}\\
        -S_{1}^{\dag} & C_{1}^{\dag}
      \end{pmatrix}
      \\
      & =
      \begin{pmatrix}
        C_{0}L_{0}C_{0}^{\dag} +
        S_{1}L_{1}S_{1}^{\dag} &
        C_{0}L_{0}S_{0}^{\dag} -
        S_{1}L_{1}C_{1}^{\dag} \\
        S_{0}L_{0}C_{0}^{\dag} -
        C_{1}L_{1}S_{1}^{\dag} &
        S_{0}L_{0}S_{0}^{\dag} +
        C_{1}L_{1}C_{1}^{\dag}
      \end{pmatrix},
    \end{split}
  \end{equation*}
  and hence,
  \begin{equation*}
    \begin{split}
      \lim_{t \to \infty}\L_{f}(Y^{(t)}, \Lambda^{(t)}) & = \lim_{t \to \infty}
      \begin{pmatrix}
        C_{0}L_{0}C_{0} & 0\\
        0 & C_{1}L_{1}C_{1}
      \end{pmatrix}\\
      & = \L_{f}(Y_{\lint}, \Lambda_{\lint(Y)}) \oplus Y_{\lbd}\\
      & = \L_{f}(Y, \Lambda),
    \end{split}
  \end{equation*}
  which completes the proof.
\end{proof}

\begin{lemma}\label{lem:div-proj-diff}
  Suppose $f$ is a Legendre convex function.
  For all $X, Y \in \Herm_{\Delta}(\X)$ such that $X \triangleright Y$, and all
  admissible $\Lambda$, the following two identities hold
  \begin{align}
    \label{eq:div-proj-diff-1}
    D_{f}(X, Y) - D_{f}(X, \L_{f}(Y, \Lambda))
    & = \ip{\Lambda}{X} - \ell_{f}(Y, \Lambda),\\
    \label{eq:div-proj-diff-2}
    D_{f}(X, Y) - D_{f}(X, \L_{f}(Y, \Lambda))
    & = D_{f}(\L_{f}(Y, \Lambda), Y) + \ip{\Lambda}{X - \L_{f}(Y, \Lambda)}.
  \end{align}
\end{lemma}

\begin{proof}
  We first prove the lemma for the case when $\supp(Y) = \X$.
  By expanding the definition of $D_{f}$ and abbreviating $\L_{f}(Y, \Lambda)$
  as $L$, we have
  \begin{align}
    & D_{f}(X, Y) - D_{f}(X, L)\nonumber\\
    =\, & \tr \bigl( f(L) - f(Y) +
          f'(L)(X - L) - f'(Y)(X - Y) \bigr)\nonumber\\
    =\, & \tr \bigl( f(L) - f(Y) +
          (f'(Y) + \Lambda)(X - L) - f'(Y)(X - Y)\bigr)\tag{\cref{pro:leg-dif}}\\
    =\, & \tr \bigl( \Lambda(X - L) +
          f(L) - f(Y) - f'(Y)(L - Y) \bigr)\nonumber\\
    =\, & \tr \bigl( \Lambda(X - L) + D_{f}(L, Y)\bigr)\nonumber\\
    =\, & \ip{\Lambda}{X} - \ell_{f}(Y, \Lambda).\label{eq:div-proj-diff-3}
  \end{align}

  In general, the definition of $D_{f}$ and $\L_{f}$ on the extended domain
  implies
  \begin{equation*}
    D_{f}(X, Y) - D_{f}(X, \L_{f}(Y, \Lambda))
    = D_{f}(\hat{X}, \hat{Y}) -
    D_{f}(\hat{X}, \L_{f}(\hat{Y}, \hat{\Lambda})),
  \end{equation*}
  where $\hat{X}$, $\hat{Y}$, and $\hat{\Lambda}$ are restrictions of $X$, $Y$,
  and $\Lambda$ to $\supp(Y)$.
  By the calculation in \cref{eq:div-proj-diff-3}, this further simplifies to
  \begin{equation*}
    \bigl\langle \hat{\Lambda}, \hat{X} \bigr\rangle
    - \ell_{f}(\hat{Y}, \hat{\Lambda}) =
    \bigl\langle \hat{\Lambda}, \hat{X} \bigr\rangle
    - \ell_{f}(Y, \Lambda) + \ip{\Lambda}{Y_{\lbd}} =
    \ip{\Lambda}{X} - \ell_{f}(Y, \Lambda).
  \end{equation*}
  This proves \cref{eq:div-proj-diff-1}.
  \Cref{eq:div-proj-diff-2} follows from \cref{eq:div-proj-diff-1} and the fact
  that $\L_{f}(Y, \Lambda)$ is the minimizer for $\ell_{f}(Y, \Lambda)$.
\end{proof}

\begin{lemma}\label{lem:div-leg-dif}
  For Legendre convex function $f$ and $X, Y \in \Herm_{\Delta}(\X)$ such that
  $X \triangleright Y$, the mapping $t \mapsto D_{f}(X, \L_{f}(Y, t \Lambda))$
  is differentiable at $t = 0$ with derivative
  \begin{equation}
    \label{eq:div-leg-dif}
    \odv*{D_{f}(X, \L_{f}(Y, t \Lambda))}{t} \at[\Big]{t=0}
    = \ip{\Lambda}{Y-X}.
  \end{equation}
\end{lemma}

\begin{proof}
  We first compute the derivative
  \begin{align}
    & \odv*{D_{f}(\L_{f}(Y, t \Lambda), Y)}{t} \at[\Big]{t=0}\nonumber\\
    =\, & \odv*{\tr \Bigl( f(\L_{f}(Y, t \Lambda)) - f'(Y) \L_{f}(Y, t \Lambda)
          \Bigr)}{t} \at[\Big]{t=0}\nonumber\\
    =\, & \ip{f'(\L_{f}(Y, t \Lambda)) - f'(Y)}{\odv*{\L_{f}(Y, t \Lambda)}{t}}
          \at[\bigg]{t=0}\nonumber\\
    =\, & 0,\label{eq:div-proj-0}
  \end{align}
  where the third line uses \cref{eq:dif-trace} and the last line follows from
  $\L_{f}(Y, 0) = Y$.
  Using \cref{eq:div-proj-diff-2} of \cref{lem:div-proj-diff}, we have
  \begin{align*}
    & \odv*{D_{f}(X, \L_{f}(Y, t \Lambda))}{t} \at[\Big]{t=0}\\
    =\, & \odv* {\Bigl(
          \ip{t \Lambda}{\L_{f}(Y, t \Lambda) - X} -
          D_{f}(\L_{f}(Y, t \Lambda), Y) \Bigr)}{t} \at[\Big]{t=0}\\
    =\, & \Bigl( \ip{\Lambda}{\L_{f}(Y, t \Lambda) - X} -
          \odv*{D_{f}(\L_{f}(Y, t \Lambda), Y)}{t}\Bigr) \at[\Big]{t=0}\\
    =\, & \ip{\Lambda}{Y - X}.\tag{By \cref{eq:div-proj-0}}
  \end{align*}
\end{proof}

\begin{lemma}\label{lem:conjugate-explicit}
  Let $f$ be a Legendre function and $\Lambda$ is admissible.
  $\ell_{f}(Y, \Lambda)$ has the following explicit form
  \begin{equation}
    \label{eq:conjugate-explicit}
    \ell_{f}(Y, \Lambda) = \tr f^{*}(f'(Y) + \Lambda) - \tr f^{*}(f'(Y)).
  \end{equation}
\end{lemma}

\begin{proof}
  Taking $X = Y$ in \cref{eq:div-proj-diff-1} and letting $R = f'(Y) + \Lambda$
  and $L = \L_{f}(Y, \Lambda) = (f^{*})'(R)$, we have
  \begin{equation*}
    \begin{split}
      \ell_{f}(Y, \Lambda) & = \ip{\Lambda}{Y} + D_{f}(Y, \L_{f}(Y, \Lambda))\\
      & = \ip{\Lambda}{Y} + \tr \bigl( f(Y) - f(L) - f'(L)(Y-L) \bigr)\\
      & = \ip{\Lambda}{Y} + \tr \bigl( f(Y) -f((f^{*})'(R)) - R\, (Y-L)\bigr)\\
      & = \tr \bigl( R\, (f^{*})'(R) - f((f^{*})'(R)) \bigr)
      - \tr \bigl( f'(Y)\,Y - f(Y)\bigr)\\
      & = \tr f^{*}(f'(Y) + \Lambda) - \tr f^{*}(f'(Y)),
    \end{split}
  \end{equation*}
  where the last line follows from \cref{lem:conjugate-form}.
\end{proof}

\begin{lemma}\label{lem:conjugate-derivative}
  For Legendre convex function $f$ and a list of Hermitian matrices
  $F = (F_{1}, F_{2}, \ldots, F_{k})$,
  \begin{equation*}
    \pdv*{\ell_{f}(Y, \lambda \cdot F)}{\lambda_{j}}
    = \ip{F_{j}}{\L_{f}(Y, \lambda \cdot F)}.
  \end{equation*}
\end{lemma}

\begin{proof}
  By \cref{lem:conjugate-explicit}, we can write
  \begin{align}
    \pdv*{ \ell_{f}(Y, \lambda \cdot F)}{\lambda_{j}}
    & = \pdv*{\tr f^{*}(f'(Y) + \lambda \cdot F)}{\lambda_{j}} \notag\\
    & = \ip{F_{j}}{(f^{*})'(f'(Y) + \lambda \cdot F)}
          \tag{By \cref{eq:dif-trace}}\\
    & = \ip{F_{j}}{\L_{f}(Y, \lambda \cdot F)}.\notag\qedhere
  \end{align}
\end{proof}

\begin{lemma}\label{lem:bounded-has-limit}
  Suppose $f$ is a Legendre convex function with domain $\Delta$ and the domain
  $\dom{f^{*}}$ is open.
  Let $c$ be a constant and a sequence $\bigl( Y^{(t)} \bigr)$ in
  $\Herm_{\Delta_{\lint}}(\X)$ satisfies that
  $D_{f} \bigl( X, Y^{(t)} \bigr) \le c$.
  Then $\bigl (Y^{(t)} \bigr)$ has a limiting point $\hat{Y}$ in
  $\Herm_{\Delta}(\X)$.
\end{lemma}

\begin{proof}
  As the domain $\dom{f^{*}}$ is open, \cref{lem:coercive} guarantees that
  $D_{f}(X, \;\cdot\;)$ is coercive for all $X \in \Herm_{\Delta}(\X)$.
  That is, the sequence $\bigl (Y^{(t)} \bigr)$ is bounded.
  Hence, there must be a subsequence $\bigl( Y^{(t_{i})} \bigr)$ of
  $\bigl( Y^{(t)} \bigr)$ that converges to $\hat{Y}$.
  We will show that $\hat{Y}$ is in $\Herm_{\Delta}(\X)$.
  Assume to the contrary that $\hat{Y}$ is not in $\Herm_{\Delta}(\X)$ and it
  has eigenvalues in $\cl{\Delta} \setminus \Delta$.
  Let $a \in \real$ be such an eigenvalue.
  $Y^{(t_{i})} = \sum_{j} y^{(i)}_{j} \bigket{\psi^{(i)}_{j}}\bigbra{\psi^{(i)}_{j}}$
  be the spectrum decomposition of $Y^{t_{i}}$.
  Define $x^{(i)}_{j} = \bigbra{\psi^{(i)}_{j}} X \bigket{\psi^{(i)}_{j}}$ and
  \begin{equation*}
    X^{(i)} = \sum_{j} x^{(i)}_{j} \bigket{\psi^{(i)}_{j}}\bigbra{\psi^{(i)}_{j}}.
  \end{equation*}
  As in the proof of \cref{lem:coercive}, we have
  \begin{equation*}
    D_{f}\bigl( X, Y^{(t_{i})} \bigr) = \tr f(X) - \tr f \bigl( X^{(i)} \bigr) +
    \sum_{j} D_{f} \bigl( x^{(i)}_{j}, y^{(i)}_{j} \bigr).
  \end{equation*}
  Then we have $D_{f} \bigl( X, Y^{(t_{i})} \bigr)$ goes to $+\infty$ by
  \cref{lem:boundary-effect}, a contradiction with the assumption that
  $D_{f} \bigl( X, Y^{(t)} \bigr) \le c$.
\end{proof}

\subsection{Linear Families and Legendre-Bregman Projection Families}%
\label{sec:families}

\begin{definition}
  Given $X_{0} \in \Herm_{\Delta}(\X)$ and a tuple $F={(F_{j})}_{j=1}^{k}$ of
  Hermitian matrices $F_{j}$, the \emph{linear family} for $X_{0}$ and $F$ is
  defined by
  \begin{equation*}
    \lf(X_{0}, F) = \bigl\{ X \in \Herm_{\Delta}(\X) \mid \ip{F_{j}}{X} =
      \ip{F_{j}}{X_{0}}, j=1, 2, \ldots, k \bigr\}.
  \end{equation*}
  For $Y_{0} \in \Herm_{\Delta}$, the \emph{Legendre-Bregman projection family}
  for $Y_{0}$ and $F$ is defined by
  \begin{equation*}
    \pf(Y_{0}, F) = \bigl\{ Y \in \Herm_{\Delta}(\X) \mid Y =
    \L_{f}(Y_{0}, \lambda \cdot F) \bigr\}.
  \end{equation*}
\end{definition}
We usually simply denote the two families of operators as $\lf$ and $\pf$ when
$X_{0}$, $Y_{0}$, and $F$ is obvious from the context.

\begin{lemma}\label{lem:nonempty}
  Suppose $f$ is a real convex function of Legendre type and $\dom{f^{*}}$ is
  open.
  If $D_{f}(X_{0}, Y_{0}) < \infty$, then $\lf \cap \clpf$ is nonempty.
\end{lemma}

\begin{proof}
  By \cref{lem:coercive}, $D_{f}(X_{0}, \;\cdot\;)$ is coercive, meaning that
  \begin{equation*}
    \mathscr{R} = \bigl\{ Y \mid D_{f}(X_{0}, Y) \le D_{f}(X_{0}, Y_{0}) \bigr\}
  \end{equation*}
  is bounded.
  Hence, the minimization
  \begin{equation*}
    \argmin\limits_{Y \in \clpf} D_{f}(X_{0}, Y) =
    \argmin\limits_{Y \in \cl{\pf \cap \mathscr{R}}} D_{f}(X_{0}, Y),
  \end{equation*}
  is obtained at a point $Y^{*}$ (not necessarily unique) in
  $\cl{\pf \cap \mathscr{R}} \subseteq \clpf$.
  We will prove that $Y^{*} \in \lf$.

  Let $\overline{Y} \in \clpf$ be such that
  $\overline{Y} = \lim_{j \to \infty} \L_{f}(Y_{0}, \mu_{j} \cdot F)$.
  Then by the continuity of the $\L_{f}$ proved in \cref{lem:proj-conti},
  \begin{equation*}
    \begin{split}
      \L_{f}(\overline{Y}, \lambda \cdot F)
      & = \lim_{j \to \infty} \L_{f}(\L_{f}(Y_{0}, \mu_{j} \cdot F),
      \lambda \cdot F)\\
      & = \lim_{j \to \infty} \L_{f}(Y_{0},
      (\mu_{j} + \lambda) \cdot F) \in \clpf.
    \end{split}
  \end{equation*}
  Thus, for the limiting point $Y^{*} \in \clpf$,
  $\L_{f}(Y^{*}, \lambda \cdot F)$ is in $\clpf$ for all admissible $\lambda$.
  By the optimality of $Y^{*}$, $D_{f}(X_{0}, \L_{f}(Y^{*}, \lambda \cdot F))$
  achieves a minimum at $\lambda = 0$.
  By \cref{lem:div-leg-dif}, we conclude that
  $\ip{F_{j}}{Y^{*}} = \ip{F_{j}}{X_{0}}$ and $Y^{*} \in \lf$.

\end{proof}

\begin{lemma}\label{lem:pythagorean}
  The Pythagorean identity
  \begin{equation}
    \label{eq:pythagorean}
    D_{f}(X, Y) = D_{f}(X, Y^{*}) + D_{f}(Y^{*}, Y)
  \end{equation}
  holds for all $X \in \lf$, $Y \in \clpf$ and $Y^{*} \in \lf \cap \clpf$.
\end{lemma}

\begin{proof}
  Suppose that $X_{1}, X_{2}, Y_{1}, Y_{2} \in \Herm_{\Delta}(\X)$ satisfying
  $\X_{j} \triangleright Y_{j}$ for $j=1, 2$ and
  $Y_{2} = \L_{f}(Y_{1}, \lambda \cdot F)$.
  By \cref{eq:div-proj-diff-1}, we have
  \begin{equation*}
    D_{f}(X_{1}, Y_{1}) - D_{f}(X_{1}, Y_{2}) = \ip{\lambda \cdot F}{X_{1}} -
    \ell_{f}(Y_{1}, \lambda \cdot F)
  \end{equation*}
  and
  \begin{equation*}
    D_{f}(X_{2}, Y_{1}) - D_{f}(X_{2}, Y_{2}) = \ip{\lambda \cdot F}{X_{2}} -
    \ell_{f}(Y_{1}, \lambda \cdot F)
  \end{equation*}
  Taking the difference of the above two equations, we have
  \begin{equation*}
    D_{f}(X_{1}, Y_{1}) - D_{f}(X_{2}, Y_{1}) - D_{f}(X_{1}, Y_{2}) +
    D_{f}(X_{2}, Y_{2}) = 0.
  \end{equation*}
  The lemma follows by choosing $X_{1} = Y_{1} = Y^{*}$.

\end{proof}

\subsection{Duality Theorem}\label{sec:duality}

\begin{theorem}\label{thm:duality}
  Suppose $f$ is a real convex function of Legendre type and $\dom{f^{*}}$ is
  open.
  Let $\Delta$ be the domain of $f$ and $X_{0}, Y_{0} \in \Herm_{\Delta}(\X)$ be
  two Hermitian operators satisfying $D_{f}(X_{0}, Y_{0}) < \infty$.
  Let $\lf$ and $\pf$ be the linear family of $X_{0}$ and the Legendre-Bregman
  projection family of $Y_{0}$ with respect to $F = (F_{j})$.
  Then there is a unique $Y^{*} \in \Herm_{\Delta}(\X)$ such that
  \begin{enumerate}
    \item\label{enu:in} $Y^{*} \in \lf \cap \clpf$,
    \item\label{enu:pythagorean} $D_{f}(X, Y) = D_{f}(X, Y^{*}) + D_{f}(Y^{*}, Y)$
      for any $X \in \lf$ and $Y \in \clpf$,
    \item\label{enu:lfp} $Y^{*} = \argmin\limits_{X \in \lf} D_{f}(X, Y_{0})$,
    \item\label{enu:pfp} $Y^{*} = \argmin\limits_{Y \in \clpf} D_{f}(X_{0}, Y)$.
  \end{enumerate}
  Moreover any one of these four conditions determines $Y^{*}$ uniquely.
\end{theorem}

\begin{proof}
  Choose any point $Y^{*}$ in $\lf \cap \clpf$, whose existence is guaranteed by
  \cref{lem:nonempty}.
  It satisfies \cref{enu:in} by definition, \cref{enu:pythagorean} by
  \cref{lem:pythagorean}.
  As a consequence of \cref{enu:pythagorean}, it also satisfies
  \cref{enu:lfp,enu:pfp}.
  More specifically, \cref{enu:lfp} holds because for all $X \in \lf$, it follows
  from \cref{enu:pythagorean} that
  \begin{equation*}
    D_{f}(X, Y_{0}) = D_{f}(X, Y^{*}) + D_{f}(Y^{*}, Y_{0}) \ge D_{f}(Y^{*}, Y_{0}),
  \end{equation*}
  and equality holds if and only if $D_{f}(X, Y^{*}) = 0$, that is $X = Y^{*}$.
  Similarly, for all $Y \in \clpf$, \cref{enu:pythagorean} implies that
  \begin{equation*}
    D_{f}(X_{0}, Y) = D_{f}(X_{0}, Y^{*}) + D_{f}(Y^{*}, Y) \ge D_{f}(X_{0}, Y^{*}),
  \end{equation*}
  and equality holds if and only if $D_{f}(Y^{*}, Y) = 0$, or equivalently
  $Y = Y^{*}$. This proves \cref{enu:pfp}.

  It remains to show that each of the four
  \cref{enu:in,enu:pythagorean,enu:lfp,enu:pfp} determines $Y^{*}$ uniquely.
  In other words, if $\hat{Y}$ is an operator in $\Herm_{\Delta}(\X)$
  satisfying any of the four Items, then $\hat{Y} = Y^{*}$.
  Suppose first $\hat{Y} \in \lf \cap \clpf$.
  It follows from \cref{enu:pythagorean} that
  \begin{equation*}
    D_{f}(\hat{Y}, Y^{*}) + D_{f}(Y^{*}, \hat{Y}) =
    D_{f}(\hat{Y}, \hat{Y}) = 0,
  \end{equation*}
  which guarantees that $\hat{Y} = Y^{*}$.
  If $\hat{Y}$ satisfies \cref{enu:pythagorean}, the same argument with the role
  of $\hat{Y}$ and $Y^{*}$ reversed proves that $\hat{Y} = Y^{*}$.
  If $\hat{Y}$ is a minimizer for \cref{enu:lfp}, we have by
  \cref{enu:pythagorean} that
  \begin{equation*}
    D_{f}(X_{0}, Y^{*}) \ge D_{f}(X_{0}, \hat{Y}) =
    D_{f}(X_{0}, Y^{*}) + D_{f}(Y^{*}, \hat{Y}),
  \end{equation*}
  which implies $D_{f}(Y^{*}, \hat{Y}) \le 0$ and $\hat{Y} = Y^{*}$.
  Similarly, if $\hat{Y}$ is a minimizer for \cref{enu:pfp}, we have by
  \cref{enu:pythagorean} that
  \begin{equation*}
    D_{f}(Y^{*}, Y_{0}) \ge D_{f}(\hat{Y}, Y_{0}) =
    D_{f}(\hat{Y}, Y^{*}) + D_{f}(Y^{*}, Y_{0}).
  \end{equation*}
  This implies that $D_{f}(\hat{Y}, Y^{*}) \le 0$ and $\hat{Y} = Y^{*}$.
  \qedhere
\end{proof}

\section{Iterative Algorithms}\label{sec:algorithms}

In this section, we present several iterative algorithms that are based on
Bregman's projection method.
We start with an exact projection algorithm given in \cref{alg:bregman}.
The fact this is indeed a Bregman projection algorithm will become obvious when
we analyze the algorithm.

\subsection{Exact Bregman Projection Algorithm}

\begin{metaalgorithm}[hbt!]
  \centering
  \begin{algspec}
    \textbf{Require:} $f$ is an extended real function with domain
    $\Delta \subseteq \real$ and $\dom{f^{*}}$ is open.
    $X_{0}, Y_{0} \in \Herm_{\Delta}(\X)$ such that
    $D_{f}(X_{0}, Y_{0}) < \infty$.\\
    \textbf{Input:} $F=(F_{1}, F_{2}, \ldots, F_{k}) \in {\Herm(\X)}^{k}$.\\
    \textbf{Output:} $\lambda^{(1)}, \lambda^{(2)}, \cdots$ such that
    \begin{equation*}
      \lim_{t\to\infty} D_{f} \bigl( X_{0}, \L_{f}(Y_{0}, \lambda^{(t)} \cdot F)
      \bigr) = \inf_{\lambda \in \real^{k}} D_{f} \bigl( X_{0}, \L_{f}(Y_{0},
      \lambda \cdot F) \bigr).
    \end{equation*}
    \begin{algorithmic}[1]
      \State{Initialize $\lambda^{(1)} = (0, 0, \ldots, 0)$.} %
      \For{$t = 1, 2, \ldots, $} %
      \State{Compute $Y^{(t)} = \L_{f}(Y_{0}, \lambda^{(t)} \cdot F)$.} %
      \State{Compute $j_{t} = \argmax_{j} \abs{\ip{F_{j}}{Y^{(t)} - X_{0}}}$.} %
      \State{Compute the unique solution for $\delta^{(t)}_{j} \in \real$ of
        equation\label{al:bregman-eq}
        \begin{equation*}
          \ip{F_{j_{t}}}{\L_{f}(Y^{(t)},  \delta^{(t)} F_{j_{t}}) - X_{0}} = 0.
        \end{equation*}} %
      \State{Update parameter
        $\lambda^{(t+1)}_{j} =
        \begin{cases}
          \lambda^{(t)}_{j} + \delta^{(t)}_{j} & \text{if } j = j_{t},\\
          \lambda^{(t)}_{j} & \text{otherwise.}
        \end{cases}
        $} %
      \EndFor{}
    \end{algorithmic}
  \end{algspec}
  \caption{Bregman's Exact Iterative Projection Algorithm.}
  \label{alg:bregman}
\end{metaalgorithm}

\begin{theorem}\label{thm:bregman-algorithm}
  \Cref{alg:bregman} outputs a sequence $\lambda^{(1)}$, $\lambda^{(2)}$,
  $\ldots$ such that
  \begin{equation*}
    \lim_{t\to\infty} D_{f} \bigl( X_{0}, \L_{f}(Y_{0}, \lambda^{(t)} \cdot F)
    \bigr) = \inf_{\lambda \in \real^{k}} D_{f} \bigl( X_{0}, \L_{f}(Y_{0},
    \lambda \cdot F) \bigr).
  \end{equation*}
\end{theorem}

\begin{proof}
  We consider the $t$-th iteration of the algorithm.
  The matrix $Y^{(t)}$ is our current estimate of the solution and the
  algorithms computes an update in \cref{al:bregman-eq} which can be understood
  as a Bregman's projection as follows.
  Let $j_{t}$ be the index used in the algorithm and define the linear family
  \begin{equation*}
    \lf_{j_{t}} = \bigl\{ X \in \Herm_{\Delta}(\X) \mid \ip{F_{j_{t}}}{X} =
    \ip{F_{j_{t}}}{X_{0}} \bigr\},
  \end{equation*}
  and Legendre-Bregman projection family
  \begin{equation*}
    \pf_{j_{t}} = \bigl\{ Y \in \Herm_{\Delta}(\X) \mid Y =
    \L_{f}(Y^{(t)}, \delta^{(t)} \, F_{j_{t}})\}.
  \end{equation*}

  The projection of $Y^{(t)}$ onto $\lf_{j_{t}}$ is in the intersection of $\lf_{j_{t}}$
  and $\pf_{j_{t}}$, and therefore it can be written as
  $\L_{f}(Y^{(t)}, \delta \, F_{j_{t}})$ for some real parameter $\delta$.
  By the fact that is also in the linear family $\lf_{j_{t}}$, we have that
  $\delta = \delta^{(t)}$ uniquely determined by the equation in
  \cref{al:bregman-eq} of \cref{alg:bregman}.

  Define an auxiliary function
  \begin{equation*}
    \aux(Y, \delta) = \ip{\delta F_{j^{*}}}{X_{0}} -
    \ell_{f} \bigl( Y, \delta F_{j^{*}} \bigr)
  \end{equation*}
  for $j^{*} = \argmax_{j} \abs{\ip{F_{j}}{Y - X_{0}}}$.
  As $\ell_{f} \bigl(Y, \delta F_{j^{*}} \bigr)$ is convex in $\delta$,
  $\aux(Y, \delta)$ is concave in $\delta$.
  In each step the algorithm makes progress measured by the quantity
  \begin{align}
    & D_{f}\bigl( X_{0}, Y^{(t)} \bigr) - D_{f} \bigl( X_{0}, Y^{(t+1)} \bigr)
    \notag\\
    =\; & D_{f} \bigl( X_{0}, Y^{(t)} \bigr) -
          D_{f} \bigl( X_{0}, \L_{f} \bigl( Y^{(t)},
          \delta^{(t)} F_{j_{t}} \bigr) \bigr)\notag\\
    =\; & \ip{\delta^{(t)} F_{j_{t}}}{X_{0}} -
          \ell_{f} \bigl( Y^{(t)}, \delta^{(t)} F_{j_{t}} \bigr)
          \tag{By \cref{eq:div-proj-diff-1}}\\
    =\; & \aux \bigl( Y^{(t)}, \delta^{(t)} \bigr)\notag.
  \end{align}
  It is easy to verify that $\aux(Y, 0) = 0$.
  From \cref{lem:conjugate-derivative} and the concavity of $\aux$ in $\delta$,
  the choice of $\delta^{(t)}$ in the algorithm maximizes
  $\aux\bigl( Y^{(t)}, \delta \bigr)$ over $\delta$ and we have
  \begin{equation*}
    \aux \bigl( Y^{(t)}, \delta^{(t)} \bigr)
    \ge \aux \bigl( Y^{(t)}, 0 \bigr) = 0.
  \end{equation*}
  Adding the above displayed equation together for $t=1, 2, \ldots, T-1$ gives
  \begin{equation*}
    D_{f} \bigl( X_{0}, Y^{(1)} \bigr) - D_{f} \bigl(X_{0}, Y^{(T)} \bigr)
    = \sum_{t=1}^{T-1} \aux \bigl( Y^{(t)}, \delta^{(t)} \bigr).
  \end{equation*}
  Since $D_{f}\bigl( X_{0}, Y^{(1)}\bigr) = D_{f}(X_{0}, Y_{0}) < \infty$ and
  $\aux \bigl(Y^{(t)}, \delta^{(t)} \bigr) \ge 0$ for all $t$, it follows that
  $D_{f} \bigl(X_{0}, Y^{(t)} \bigr)$ is bounded and, by
  \cref{lem:bounded-has-limit}, the sequence
  ${\bigl( Y^{(t)} \bigr)}_{t=1}^{\infty}$ has a subsequence
  $\bigl( Y^{(t_{i})} \bigr)$ converging to $\hat{Y} \in \Herm_{\Delta}(\X)$.

  We claim that for the limiting point $\hat{Y}$,
  $\max_{\delta}\aux(\hat{Y}, \delta)=0$.
  Otherwise, assume there exist a $\delta$ and an $\varepsilon > 0$ such that
  $\aux(\hat{Y}, \delta) = \varepsilon > 0$.
  By the continuity of $\aux$ with respect to $Y$, it follows that there is an
  integer $m > 0$ such that for all $i \ge m$,
  $\aux(Y^{(t_{i})}, \delta^{(t_{i})}) \ge \aux(Y^{(t_{i})}, \delta) \ge \varepsilon /2$.
  This is a contradiction with the fact that
  ${(\aux(Y^{(t)}, \delta^{(t)}))}_{t}$ converges to $0$.

  Finally, we show that the condition $\max_{\delta} \aux(\hat{Y}, \delta) = 0$
  implies that $\hat{Y}$ is in $\lf$.
  Assume on the other hand that $\hat{Y}$ is not in $\lf$ and, hence, for
  \begin{equation*}
   j^{*} = \argmax_{j} \abs{\ip{F_{j}}{\hat{Y} - X_{0}}},
  \end{equation*}
  we must have
  \begin{equation}
    \label{eq:bregman-algorithm-0}
    \ip{F_{j^{*}}}{\hat{Y} - X_{0}} \ne 0.
  \end{equation}
  As
  \begin{equation*}
    \odv*{\aux(Y,\delta)}{\delta}
    = \ip{F_{j^{*}}}{X_{0}} - \ip{F_{j^{*}}}{\L_{f}(\hat{Y}, \delta F_{j^{*}})},
  \end{equation*}
  the maximizer $\delta^{*}$ of $\max_{\delta} \aux(\hat{Y}, \delta)$ satisfies
  \begin{equation}
    \label{eq:bregman-algorithm-1}
    \ip{F_{j^{*}}}{X_{0} - \L_{f}(Y, \delta^{*} F_{j^{*}})} = 0.
  \end{equation}
  Therefore,
  \begin{equation*}
    \begin{split}
      \max_{\delta} \aux(\hat{Y}, \delta)
      & = \aux(\hat{Y}, \delta^{*})\\
      & = \ip{\delta^{*} F_{j^{*}}}{X_{0}} -
      \ell_{f}(\hat{Y}, \delta^{*} F_{j^{*}})\\
      & = \ip{\delta^{*} F_{j^{*}}}{X_{0} - \L_{f}(\hat{Y}, \delta^{*} F_{j^{*}})}
      + D_{f}(\hat{Y}, \L_{f}(\hat{Y}, \delta^{*} F_{j^{*}}))\\
      & = D_{f} \bigl( \hat{Y}, \L_{f}(\hat{Y}, \delta^{*} F_{j^{*}}) \bigr),
    \end{split}
  \end{equation*}
  where the third line uses the definition of $\ell_{f}$ and the fourth uses
  \cref{eq:bregman-algorithm-1}.
  Now, $\max_{\delta} \aux(\hat{Y}, \delta) = 0$ implies
  $D_{f} \bigl( \hat{Y}, \L_{f}(\hat{F}, \delta^{*} F_{j^{*}}) \bigr) = 0$ and
  $\hat{Y} = \L_{f}(\hat{F}, \delta^{*} F_{j^{*}})$.
  But then \cref{eq:bregman-algorithm-1} becomes
  \begin{equation*}
    \ip{F_{j^{*}}}{X_{0} - \hat{Y}} = 0,
  \end{equation*}
  contradicting \cref{eq:bregman-algorithm-0}.
  Therefore $\hat{Y} \in \lf$ and by definition $\hat{Y}$ is also in $\clpf$.
  By the duality theorem in \cref{thm:duality}, $\hat{Y}$ is the unique
  projection of $X_{0}$ to the projection family $\pf$.
\end{proof}


\subsection{Approximate Bregman Projection Algorithms}%
\label{sec:approximate-projection-algorithms}

Most of the time, the exact Bregman projection equation (\cref{al:bregman-eq} in
\cref{alg:bregman}) could be hard to solve.
In the literature, a class of approximate projection algorithms are known whose
corresponding equations are usually much easier to solve and sometimes have
simple explicit formulas.
We follow the auxiliary function approach~\cite{D-PD-PL02,CSS02} to derive our
approximate projection algorithms.
Like in~\cite{CSS02}, both a parallel update algorithm (\cref{alg:parallel}) and
a sequential update algorithm (\cref{alg:sequential}) are considered.
Thanks to the generality of the auxiliary function method, we are able to prove
the convergence for them in a uniform manner.

The equation in \cref{al:parallel-eq} of \cref{alg:parallel} may seem
complicated at first glance, but as we will see later in special examples, it is
a simpler equation to solve.
In the parallel projection algorithm, the parameters corresponding to $F_{j}$
are updated simultaneously in each iteration.
In this case, we require that $\sum_{j=1}^{k} \abs{F_{j}} \le \I$, a technical
condition for the convergence proof to work in the \emph{parallel} update
algorithms.

Unlike the classical case where the convergence of similar approximate
projection algorithms always holds, in the non-commutative case, we require a
\emph{strong convex} condition of $\ell_{f}$ to hold.

\begin{definition}\label{def:strong-convex}
  The Legendre-Bregman conjugate $\ell_{f}(Y, \Lambda)$ is strongly convex if
  for operators $F_{j} \succeq 0$ and $\sum_{j}F_{j} = \I$,
  \begin{equation*}
    \ell_{f}(Y, \delta \cdot F) \le \tr \biggl( \sum_{j=1}^{k}
    \hat{\ell}_{f}(Y, \delta_{j}) F_{j} \biggr),
  \end{equation*}
  where $\hat{\ell}_{f}(Y, \delta_{j})$ is the application of
  $\ell_{f}(\;\cdot\;, \delta_{j})$ to $Y$.
\end{definition}

\begin{lemma}
  The Legendre-Bregman conjugate is strongly convex if and only if
  \begin{equation*}
    \tr f^{*} \bigl( f'(Y) + \delta \cdot F \bigr) \le
    \tr \sum_{j=1}^{k} \Bigl( f^{*} \bigl( f'(Y) +
    \delta_{j} \bigr) F_{j} \Bigr).
  \end{equation*}
\end{lemma}

\begin{proof}
  This directly follows from \cref{lem:conjugate-explicit}.
\end{proof}

If $f'(Y)$ and $F_{j}$'s are commuting, the inequality for strong convexity is
always true and follows from Jensen's trace inequality as long as $f^{*}$ is
convex.
This is not the case however in the non-commutative case.

\begin{metaalgorithm}[hbt!]
  \centering
  \begin{algspec}
    \textbf{Require:} $f$ is an extended real function with domain
    $\Delta \subseteq \real$ and the domain $\Delta^{*}$ of $f^{*}$ is open.
    $X_{0}, Y_{0} \in \Herm_{\Delta}(\X)$ such that
    $D_{f}(X_{0}, Y_{0}) < \infty$.\\
    \textbf{Input:} $F=(F_{1}, F_{2}, \ldots, F_{k}) \in {\Herm(\X)}^{k}$
    and $\sum_{j=1}^{k} \abs{F_{j}} \le \I$.\\
    \textbf{Output:} $\lambda^{(1)}, \lambda^{(2)}, \cdots$ such that
    \begin{equation*}
      \lim_{t\to\infty} D_{f} \bigl( X_{0}, \L_{f}(Y_{0}, \lambda^{(t)} \cdot F)
      \bigr) = \inf_{\lambda \in \real^{k}} D_{f} \bigl( X_{0}, \L_{f}(Y_{0},
      \lambda \cdot F) \bigr).
    \end{equation*}
    \begin{algorithmic}[1]
      \State{Initialize $\lambda^{(1)} = (0, 0, \ldots, 0)$.} %
      \For{$t = 1, 2, \ldots, $} %
      \State{Compute $Y^{(t)} = \L_{f}(Y_{0}, \lambda^{(t)} \cdot F)$.}
      \For{$j = 1, 2, \ldots, k$} %
      \State{Solve the following equation of
        $\delta^{(t)}_{j}$:
        \begin{equation*}
          \tr \Bigl( F_{j}^{+} \L_{f} (Y^{(t)}, \delta^{(t)}_{j}) -
          F_{j}^{-} \L_{f}(Y^{(t)}, -\delta^{(t)}_{j}) \Bigr) =
          \ip{F_{j}}{X_{0}}.
        \end{equation*}\label{al:parallel-eq}} %
      \EndFor{} %
      \State{Update parameters
        $\lambda^{(t+1)} = \lambda^{(t)} + \delta^{(t)}$.} %
      \EndFor{}
    \end{algorithmic}
  \end{algspec}
  \caption{Parallel Approximate Projection Algorithm.}
  \label{alg:parallel}
\end{metaalgorithm}

\begin{theorem}\label{thm:parallel}
  Suppose $\ell_{f}$ is strongly convex.
  Then \cref{alg:parallel} outputs a sequence
  $\lambda^{(1)}, \lambda^{(2)}, \ldots$ such that
  \begin{equation*}
    \lim_{t\to\infty} D_{f} \bigl( X_{0}, \L_{f}(Y_{0}, \lambda^{(t)} \cdot F)
    \bigr) = \inf_{\lambda \in \real^{k}} D_{f} \bigl( X_{0}, \L_{f}(Y_{0},
    \lambda \cdot F) \bigr).
  \end{equation*}
\end{theorem}

\begin{proof}
  For all $\delta \in \real$ and Hermitian matrix $Y$, define
  $\hat{\ell}_{f}(Y, \delta)$ as the application of function
  $\ell_{f}(\;\cdot\;, \delta)$ to matrix $Y$.
  That is, if $Y = \sum_{i} y_{i}\, \Pi_{i}$ is a spectrum decomposition of $Y$,
  \begin{equation*}
    \hat{\ell}_{f}(Y, \delta) = \sum_{i} \ell_{f}(y_{i}, \delta) \, \Pi_{i}.
  \end{equation*}

  Define an auxiliary function $\aux$ as
  \begin{equation}
    \label{eq:aux-func}
    \aux(Y, \delta) = \ip{\delta \cdot F}{X_{0}} - \tr \sum_{j=1}^{k}
    \Bigl( \hat{\ell}_{f}(Y, \delta_{j}) F_{j}^{+} +
    \hat{\ell}_{f}(Y, - \delta_{j}) F_{j}^{-} \Bigr).
  \end{equation}

  We claim several properties of the auxiliary function $\aux$.

  First, it is easy to verify that $\aux(Y, 0) = 0$.
  This follows from the definition of $\aux$ and the fact that
  $\hat{\ell}_{f}(Y, 0) = 0$ for all $Y$.

  Second, the choice of $\delta^{(t)}_{j}$'s in the algorithm is exactly those
  so that $\aux(Y^{(t)}, \delta^{(t)})$ is maximized over $\delta^{(t)}$.
  To see this, we first notice that $\hat{\ell}_{f}(Y,\delta)$ is convex in
  $\delta$, $\aux$ is concave in $\delta^{(t)}$, and the maximum is achieved at
  the point with zero gradient.
  Next, we use \cref{lem:conjugate-derivative} to compute the derivative of
  $\aux(Y^{(t)}, \delta^{(t)})$ with respect to $\delta^{(t)}_{j}$ as
  \begin{equation*}
    \ip{F_{j}}{X_{0}} - \tr \Bigl( \L_{f}(Y^{(t)}, \delta^{(t)}_{j}) F_{j}^{+} -
    \L_{f}(Y^{(t)}, -\delta^{(t)}_{j}) F_{j}^{-} \Bigr).
  \end{equation*}
  This proves our claim about the optimality of $\delta^{(t)}_{j}$ and, in
  particular, it holds that
  $\aux(Y^{(t)}, \delta^{(t)}) \ge \aux(Y^{(t)}, 0) = 0$.

  Third, if $\max_{\delta} \aux(Y, \delta) = 0$, then $Y \in \lf$.
  Let $\delta^{*} = \argmax_{\delta} \aux(Y, \delta)$ be the maximizer of
  $\aux(Y, \delta)$.
  It satisfies the equations
  \begin{equation}
    \label{eq:parallel-1}
    \ip{F_{j}}{X_{0}} - \tr \Bigl( \L_{f}(Y, \delta^{*}_{j}) F_{j}^{+}
    - \L_{f}(Y, - \delta^{*}_{j}) F_{j}^{-} \Bigr) = 0,
  \end{equation}
  for all $j$.
  We compute
  \begin{equation*}
    \begin{split}
    \aux(Y, \delta^{*}) & = \ip{\delta^{*} \cdot F}{X_{0}} - \tr \sum_{j=1}^{k}
    \Bigl( \hat{\ell}_{f}(Y, \delta^{*}_{j}) F_{j}^{+} +
    \hat{\ell}_{f}(Y, - \delta^{*}_{j}) F_{j}^{-} \Bigr)\\
    & = \ip{\delta^{*} \cdot F}{X_{0}} - \tr \sum_{j=1}^{k} \sum_{i}
    \Bigl( \ell_{f}(y_{i}, \delta^{*}_{j}) \Pi_{i} F_{j}^{+}
    + \ell_{f}(y_{i}, -\delta^{*}_{j}) \Pi_{i} F_{j}^{-} \Bigr).
    \end{split}
  \end{equation*}
  By the definition of $\ell_{f}$, we have
  \begin{equation*}
    \ell_{f}(y_{i}, \delta^{*}_{j}) = \L_{f}(y_{i}, \delta^{*}_{j})
    \, \delta_{j}^{*} - D_{f} \bigl( \L_{f}(y_{i}, \delta^{*}_{j}), y_{i} \bigr),
  \end{equation*}
  and we can continue the calculation for $\aux(Y, \delta^{*})$ as
  \begin{equation*}
    \begin{split}
      \aux(Y, \delta^{*})
      & = \ip{\delta^{*} \cdot F}{X_{0}} - \tr \sum_{j=1}^{k} \Bigl(
      \L_{f}(Y, \delta^{*}_{j}) F_{j}^{+} \delta^{*}_{j} -
      \L_{f}(Y, -\delta^{*}_{j}) F_{j}^{-} \delta^{*}_{j} \Bigr) \\
      & \;\; + \tr \sum_{j=1}^{k} \sum_{i} \Bigl( D_{f} \bigl( \L_{f}(y_{i},
      \delta^{*}_{j}), y_{i} \bigr) \Pi_{i} F_{j}^{+} + D_{f} \bigl(
      \L_{f}(y_{i}, -\delta^{*}_{j}), y_{i} \bigr) \Pi_{i} F_{j}^{-} \Bigr)\\
      & = \tr \sum_{j=1}^{k} \sum_{i} \Bigl( D_{f} \bigl( \L_{f}(y_{i},
      \delta^{*}_{j}), y_{i} \bigr) \Pi_{i} F_{j}^{+} + D_{f} \bigl(
      \L_{f}(y_{i}, -\delta^{*}_{j}), y_{i} \bigr) \Pi_{i} F_{j}^{-} \Bigr).\\
    \end{split}
  \end{equation*}
  Therefore, $\aux(Y, \delta^{*}) = 0$ implies that
  $D_{f}(\L_{f}(y_{i}, \delta^{*}_{j}), y_{i}) = 0$ (or equivalently,
  $y_{i} = \L_{f}(y_{i}, \delta_{j}^{*})$) for all $i, j$ such that
  $\tr(\Pi_{i}F_{j}^{+}) > 0$ and
  $D_{f}(\L_{f}(y_{i}, -\delta^{*}_{j}), y_{i}) = 0$ (or
  $y_{i} = \L_{f}(y_{i}, -\delta_{j}^{*})$) for all $i, j$ such that
  $\tr(\Pi_{i} F_{j}^{-}) > 0$.
  These conditions and \cref{eq:parallel-1} guarantee that
  \begin{equation*}
    \ip{F_{j}}{X_{0} - Y} =
    \ip{F_{j}}{X_{0}} - \tr \Bigl( \L_{f}(Y, \delta^{*}_{j}) F_{j}^{+}
    - \L_{f}(Y, - \delta^{*}_{j}) F_{j}^{-} \Bigr) = 0,
  \end{equation*}
  for all $j = 1, 2, \ldots, k$ and $Y \in \lf$.

  Fourth, it is a lower bound on the improvement each iteration measured by the
  difference between
  \begin{equation*}
    D_{f}(X_{0}, Y^{(t)}) - D_{f}(X_{0}, Y^{(t+1)}) \ge \aux(Y^{(t)}, \delta^{(t)}).
  \end{equation*}
  More generally, we prove that for all $Y\in \Herm_{\Delta}(\X)$ such that
  $X_{0} \triangleright Y$, and $\delta \in \real^{k}$ such that
  $\delta \cdot F$ is admissible with respect to $Y$,
  \begin{equation}
    \label{eq:aux-bound}
    D_{f}(X_{0}, Y) - D_{f}(X_{0}, \L_{f}(Y, \delta \cdot F))
    \ge \aux(Y, \delta).
  \end{equation}
  The special condition follows from this more general form by choosing
  $Y = Y^{(t)}$ and noticing that
  $Y^{(t+1)} = \L_{f}(Y^{(t)}, \delta^{(t)} \cdot F)$.
  To prove the general bound in \cref{eq:aux-bound}, we use
  \cref{eq:div-proj-diff-1} and the strong convexity assumption on $\ell_{f}$:
  \begin{equation*}
    \begin{split}
      & D_{f}(X_{0}, Y) - D_{f}(X_{0}, \L_{f}(Y, \delta \cdot F))\\
      =\, & \ip{\delta \cdot F}{X_{0}} - \ell_{f}(Y, \delta \cdot F)\\
      \ge\, & \ip{\delta \cdot F}{X_{0}} -
      \tr \sum_{j=1}^{k} \Bigl( F_{j}^{+} \hat{\ell}_{f}(Y, \delta_{j}) +
      F_{j}^{-} \hat{\ell}_{f}(Y, - \delta_{j}) \Bigr)\\
      =\, & \aux(Y, \delta).
    \end{split}
  \end{equation*}

  This implies that for all $t=1, 2, \ldots$,
  \begin{equation*}
    D_{f}(X_{0}, Y^{(t)}) - D_{f}(X_{0}, Y^{(t+1)}) \ge \aux(Y^{(t)}, \delta^{(t)}).
  \end{equation*}
  Adding these inequalities together, we have for all $T \ge 1$,
  \begin{equation*}
    D_{f}(X_{0}, Y^{(1)}) - D_{f}(X_{0}, Y^{(T)}) \ge
    \sum_{t=1}^{T-1} \aux(Y^{(t)}, \delta^{(t)}).
  \end{equation*}
  As $D_{f}(X_{0}, Y^{(1)}) = D_{f}(X_{0}, Y_{0}) < \infty$, and both $D_{f}$
  and $\aux$ are nonnegative, the sequence
  \begin{equation*}
    {(\aux(Y^{(t)}, \delta^{(t)}))}_{t}
  \end{equation*}
  converges to $0$.

  By \cref{lem:coercive} and the fact that $D_{f}(X_{0}, Y^{(t)})$ is a
  non-increasing sequence, it follows that the sequence $Y^{(t)}$ is in a
  compact subset of $\Herm_{\Delta}(\X)$ and that there is subsequence
  $Y^{(t_{i})}$ converging to a limiting point $\hat{Y}$.

  We claim that $\max_{\delta}\aux(\hat{Y}, \delta)=0$.
  Otherwise, assume there exist a $\delta$ and an $\varepsilon > 0$ such that
  $\aux(\hat{Y}, \delta) = \varepsilon > 0$.
  By the continuity of $\aux$ with respect to $Y$, it follows that there is an
  integer $m > 0$ such that for all $i \ge m$,
  \begin{equation*}
    \aux \bigl( Y^{(t_{i})}, \delta^{(t_{i})} \bigr)
    \ge \aux \bigl(Y^{(t_{i})}, \delta \bigr) \ge \varepsilon /2.
  \end{equation*}
  This is a contradiction with the fact that
  ${\bigl( \aux \bigl( Y^{(t)}, \delta^{(t)} \bigr) \bigr)}_{t}$ converges to
  $0$.
  Hence the assumption is false and
  we conclude that $\max_{\delta}\aux(\hat{Y}, \delta) = 0$.

  This now implies that $\hat{Y}$ is the intersection of $\lf$ and $\pf$ and we
  complete the proof using \cref{thm:duality}.
\end{proof}

\begin{metaalgorithm}[hbt!]
  \centering
  \begin{algspec}
    \textbf{Require:} $f$ is an extended real function with domain
    $\Delta \subseteq \real$ and domain of $f^{*}$ is open.
    $X_{0}, Y_{0} \in \Herm_{\Delta}(\X)$ such that
    $D_{f}(X_{0}, Y_{0}) < \infty$.\\
    \textbf{Input:} $F=(F_{1}, F_{2}, \ldots, F_{k}) \in {\Herm(\X)}^{k}$
    and $\abs{F_{j}} \le \I$ for all $j$.\\
    \textbf{Output:} $\lambda^{(1)}, \lambda^{(2)}, \cdots$ such that
    \begin{equation*}
      \lim_{t\to\infty} D_{f} \bigl( X_{0}, \L_{f}(Y_{0}, \lambda^{(t)} \cdot F)
      \bigr) = \inf_{\lambda \in \real^{k}} D_{f} \bigl( X_{0}, \L_{f}(Y_{0},
      \lambda \cdot F) \bigr).
    \end{equation*}
    \begin{algorithmic}[1]
      \State{Initialize $\lambda^{(1)} = (0, 0, \ldots, 0)$.} %
      \For{$t = 1, 2, \ldots, $} %
      \State{Compute $Y^{(t)} = \L_{f}(Y_{0}, \lambda^{(t)} \cdot F)$.}
      \State{Compute $j_{t} = \argmax_{j} \abs{\ip{F_{j}}{Y^{(t)} - X_{0}}}$.}
      \State{Solve the following equation of
        $\delta^{(t)}_{j} \in \real$:
        \begin{equation*}
          \tr \Bigl( F_{j_{t}}^{+} \L_{f} (Y^{(t)}, \delta^{(t)}_{j}) -
          F_{j_{t}}^{-} \L_{f}(Y^{(t)}, -\delta^{(t)}_{j}) \Bigr) =
          \ip{F_{j_{t}}}{X_{0}}.
        \end{equation*}}\label{al:sequential-eq}%
      \State{Update parameters
        $\lambda^{(t+1)}_{j} =
        \begin{cases}
          \lambda^{(t)}_{j} + \delta^{(t)}_{j} & \text{if } j = j_{t},\\
          \lambda^{(t)}_{j} & \text{otherwise.}
        \end{cases}
        $}
      \EndFor{}
    \end{algorithmic}
  \end{algspec}
  \caption{Sequential Approximate Projection Algorithm.}
  \label{alg:sequential}
\end{metaalgorithm}

\begin{theorem}\label{thm:sequential}
  Suppose $\ell_{f}$ is strongly convex.
  Then \cref{alg:sequential} outputs a sequence
  $\lambda^{(1)}, \lambda^{(2)}, \ldots$ such that
  \begin{equation*}
    \lim_{t\to\infty} D_{f} \bigl( X_{0}, \L_{f}(Y_{0}, \lambda^{(t)} \cdot F)
    \bigr) = \inf_{\lambda \in \real^{k}} D_{f} \bigl( X_{0}, \L_{f}(Y_{0},
    \lambda \cdot F) \bigr).
  \end{equation*}
\end{theorem}

\begin{proof}
  Thanks to the flexibility of the auxiliary function proof technique, the
  analysis is identical to that of \cref{thm:parallel} by choosing
  $\delta^{(t)}_{l} = 0$ for all $l \ne j_{t}$ in the $t$-th iteration.
\end{proof}



\section{Approximate Information Projection Algorithms}\label{sec:ip}

In this section, we discuss several interesting special cases of the general
framework when the convex function $f$ is $x\ln(x)-x$.
For such a convex function, the Bregman divergence is known as Kullback–Leibler
divergence and the Bregman projection is also known as the information
projection for the central role Kullback–Leibler divergence plays in information
theory.
We will show that many important classical algorithms in learning theory
generalize to the quantum (non-commutative) case nicely in this framework.

\subsection{General Approximate Information Projection Algorithms}

For $f(x)=x\ln(x)-x$, we compute the relevant functions and quantities as
follows
\begin{equation}
  \label{eq:kl}
  \begin{aligned}
    f^{*}(x) & = \exp(x),\\
    f'(x) & = \ln(x),\\
    D_{f}(X, Y) & = \tr \bigl(X \ln X - X \ln Y - X + Y\bigr),\\
    \L_{f}(Y, \Lambda) & = \exp(\ln(Y) + \Lambda),\\
    \ell_{f}(Y, \Lambda) & = \tr \exp(\ln(Y) + \Lambda) - \tr Y.
  \end{aligned}
\end{equation}
Because of the fundamental importance of this case, we may sometimes omit the
subscript $f$ and use $D(X, Y)$, $\L(Y, \Lambda)$, and $\ell(Y, \Lambda)$ to
denote $D_{f}(X, Y)$, $\L_{f}(Y, \Lambda)$, $\ell_{f}(Y, \Lambda)$ respectively.
For density matrices $\rho$ and $\sigma$, $D(\rho, \sigma)$ reduces to the usual
Kullback–Leibler divergence
$D(\rho \Vert \sigma) = \tr (\rho \ln \rho - \rho \ln \sigma)$.
We choose to work with the non-normalized Kullback-Leibler divergence as it is
more flexible.

\begin{algorithm}[hbt!]
  \centering
  \begin{algspec}
    \textbf{Require:}
    $X_{0}, Y_{0} \in \Herm_{\Delta}(\X)$ such that $D(X_{0}, Y_{0}) < \infty$.\\
    \textbf{Input:} $F=(F_{1}, F_{2}, \ldots, F_{k}) \in {\Herm(\X)}^{k}$
    and $\sum_{j=1}^{k} \abs{F_{j}} \le \I$.\\
    \textbf{Output:} $\lambda^{(1)}, \lambda^{(2)}, \cdots$ such that
    \begin{equation*}
      \lim_{t\to\infty} D \bigl( X_{0}, \L(Y_{0}, \lambda^{(t)} \cdot F)
      \bigr) = \inf_{\lambda \in \real^{k}} D \bigl( X_{0}, \L(Y_{0},
      \lambda \cdot F) \bigr).
    \end{equation*}
    \begin{algorithmic}[1]
      \State{Initialize $\lambda^{(1)} = (0, 0, \ldots, 0)$.} %
      \For{$t = 1, 2, \ldots $} %
      \State{Compute $Y^{(t)} = \exp(\ln Y_{0} + \lambda^{(t)} \cdot F)$.}
      \For{$j = 1, 2, \ldots, k$} %
      \State{Solve the following equation of
        $\delta^{(t)}_{j}$:\label{al:parallel-kl-eq}
        \begin{equation*}
          \ip{F_{j}^{+}}{Y^{(t)}} \, \e^{\,\delta^{(t)}_{j}} -
          \ip{F_{j}^{-}}{Y^{(t)}} \, \e^{-\delta^{(t)}_{j}} =
          \ip{F_{j}}{X_{0}}.
        \end{equation*}} %
      \EndFor{} %
      \State{Update parameters
        $\lambda^{(t+1)} = \lambda^{(t)} + \delta^{(t)}$.} %
      \EndFor{}
    \end{algorithmic}
  \end{algspec}
  \caption{Parallel Iterative Update Algorithm for Kullback-Leibler Divergence
    Minimization.}
  \label{alg:parallel-kl}
\end{algorithm}

\begin{algorithm}[hbt!]
  \centering
  \begin{algspec}
    \textbf{Require:}
    $X_{0}, Y_{0} \in \Herm_{\Delta}(\X)$ such that $D(X_{0}, Y_{0}) < \infty$.\\
    \textbf{Input:} $F=(F_{1}, F_{2}, \ldots, F_{k}) \in {\Herm(\X)}^{k}$
    and $\abs{F_{j}} \le \I$ for all $j$.\\
    \textbf{Output:} $\lambda^{(1)}, \lambda^{(2)}, \cdots$ such that
    \begin{equation*}
      \lim_{t\to\infty} D \bigl( X_{0}, \L(Y_{0}, \lambda^{(t)} \cdot F)
      \bigr) = \inf_{\lambda \in \real^{k}} D \bigl( X_{0}, \L(Y_{0},
      \lambda \cdot F) \bigr).
    \end{equation*}
    \begin{algorithmic}[1]
      \State{Initialize $\lambda^{(1)} = (0, 0, \ldots, 0)$.} %
      \For{$t = 1, 2, \ldots $} %
      \State{Compute $Y^{(t)} = \exp(\ln Y_{0} + \lambda^{(t)} \cdot F)$.}
      \State{Compute $j_{t} = \argmax_{j} \abs{\ip{F_{j}}{Y^{(t)} - X_{0}}}$.}
      \State{Compute
        $\delta^{(t)}$:\label{al:sequential-kl-eq}
        \begin{equation*}
          \ip{F_{j_{t}}^{+}}{Y^{(t)}} \, \e^{\,\delta^{(t)}_{j}} -
          \ip{F_{j_{t}}^{-}}{Y^{(t)}} \, \e^{-\delta^{(t)}_{j}} =
          \ip{F_{j_{t}}}{X_{0}}.
        \end{equation*}} %
      \State{Update parameters
        $\lambda^{(t+1)}_{j} =
        \begin{cases}
          \lambda^{(t)}_{j} + \delta^{(t)} & \text{if } j = j_{t},\\
          \lambda^{(t)}_{j} & \text{otherwise.}
        \end{cases}$} %
      \EndFor{}
    \end{algorithmic}
  \end{algspec}
  \caption{Sequential Iterative Update Algorithm for Kullback-Leibler Divergence
    Minimization.}
  \label{alg:sequential-kl}
\end{algorithm}

In this special case, it is possible to solve the equation in
\cref{al:parallel-eq} of the algorithm analytically, which turns out to be a
quadratic equation in $\e^{\, \delta^{(t)}_{j}}$.
The advantage of making \emph{approximate} Bregman projection steps is now
evident as the calculation of the update in each step is extremely simple.
In light of \cref{thm:duality}, this algorithm leads to \cref{alg:parallel-kl}
that minimizes the Kullback-Leibler divergence subject to linear constraints.
In particular, the linear constraints are given by the matrix $X_{0}$ as
$\ip{F_{j}}{X - X_{0}} = 0$ for all $j=1, 2, \ldots, k$ and the problem is to
compute the minimizer of $D(X, Y_{0})$ for $X$ satisfying these linear
constraints.
\Cref{alg:parallel-kl} is a parallel update algorithm, but as in
\cref{sec:approximate-projection-algorithms}, it is also straightforward to
define and analyze a sequential version where only one entry of the parameter
$\lambda$ is updated each step.
This is given in \cref{alg:sequential-kl}.

We now prove the convergence of \cref{alg:parallel-kl} which follows directly
from the general convergence theorem in \cref{thm:parallel} as long as we can
prove that $\ell$ is strongly convex.

\begin{lemma}\label{lem:ell-strong-convex}
  Function $\ell(Y, \Lambda) = \tr\exp(\ln (Y) + \Lambda) - \tr Y$ is strongly
  convex as defined in \cref{def:strong-convex}.
\end{lemma}

\begin{proof}
  Let $F_{j} \succeq 0$ be positive semi-definite matrices satisfying
  $\sum_{j=1}^{k} F_{i} \le \I$, and let $Y$ be positive.
  $\delta_{j} \in \real$ are real numbers.
  By definition, we need to show that
  \begin{equation}
    \label{eq:ell-strong-convex-1}
    \ell \biggl(Y, \sum_{j=1}^{k} \delta_{j} F_{j} \biggr) \le
    \sum_{j=1}^{k} \tr \Bigl( \hat{\ell} (Y, \delta_{j}) F_{j} \Bigr).
  \end{equation}
  We first compute the right hand side as
  \begin{equation*}
    \begin{split}
      \sum_{j=1}^{k} \tr \Bigl( \hat{\ell}(Y, \delta_{j}) F_{j} \Bigr)
      & = \sum_{j=1}^{k} \sum_{i} \tr \Bigl( \ell(y_{i}, \delta_{j})
      \Pi_{i} F_{j} \Bigr)\\
      & = \sum_{j=1}^{k} \sum_{i} \tr \Bigl( \bigl( e^{\ln(y_{i}) + \delta_{j}}
      - y_{i} \bigr) \Pi_{i} F_{j} \Bigr)\\
      & = \tr \biggl( Y \sum_{j=1}^{k}e^{\delta_{j}} F_{j} \biggr)
      - \tr \biggl(Y \sum_{j=1}^{k} F_{j} \biggr).
    \end{split}
  \end{equation*}
  The inequality in \cref{eq:ell-strong-convex-1} then simplifies to
  \begin{equation*}
    \tr \exp \biggl( \ln(Y) + \sum_{j=1}^{k} \delta_{j} F_{j} \biggr)
    \le \tr \biggl( Y \sum_{j=1}^{k}e^{\delta_{j}} F_{j} \biggr)
      + \tr \biggl(Y \biggl( \I - \sum_{j=1}^{k} F_{j} \biggr)\biggr).
  \end{equation*}
  Define $F_{0} = \I - \sum_{j=1}^{k} F_{j}$ and $\delta_{0} = 0$.
  The above inequality can be written as
  \begin{equation}
    \label{eq:ell-strong-convex-2}
    \tr \exp \biggl( \ln(Y) + \sum_{j=0}^{k} \delta_{j} F_{j} \biggr)
    \le \tr \biggl( Y \sum_{j=0}^{k}e^{\delta_{j}} F_{j} \biggr)
  \end{equation}
  for $\sum_{j=0}^{k} F_{j} = \I$ and $F_{j} \succeq 0$.

  In the following we prove \cref{eq:ell-strong-convex-2} using
  \cref{lem:carlen-lieb} and this will complete the proof.
  By choosing $H = \sum_{j=0}^{k} \delta_{j} F_{j}$ and
  \begin{equation*}
    Q = \sum_{j=0}^{k} \e^{\delta_{j}} F_{j} \big/
    \tr \Bigl( Y \sum_{i=0}^{k} \e^{\delta_{i}} F_{i} \Bigr),
  \end{equation*}
  in the Carlen-Lieb inequality, we have
  \begin{equation}
    \begin{split}
      \label{eq:carlen-lieb-1}
      & \tr \exp \biggl( \ln(Y) + \sum_{j=0}^{k} \delta_{j} F_{j} \biggr)\\
      \le & \exp \biggl( \lambda_{\max} \biggl( \sum_{j=0}^{k} \delta_{j} F_{j} -
      \ln Q \biggr) \biggr)\\
      \le & \tr \biggl( Y \, \sum_{j=0}^{k} \e^{\delta_{j}} F_{j} \biggr)
      \exp \biggl( \lambda_{\max} \biggl( \sum_{j=0}^{k} \delta_{j} F_{j} -
      \ln \biggl( \sum_{j=0}^{k}\e^{\delta_{j}} F_{j} \biggr) \biggr) \biggr).
    \end{split}
  \end{equation}
  By the operator concavity of the $\ln$ function and the operator Jensen's
  inequality, we have
  \begin{equation*}
    \sum_{j=0}^{k} \delta_{j} F_{j} - \ln \biggl( \sum_{j=0}^{k}
    \e^{\delta_{j}} F_{j} \biggr) \preceq 0,
  \end{equation*}
  and therefore
  \begin{equation*}
    \lambda_{\max} \biggl( \sum_{j=0}^{k} \delta_{j} F_{j} -
    \ln \biggl( \sum_{j=0}^{k} \e^{\delta_{j}} F_{j} \biggr) \biggr) \le 0.
  \end{equation*}
  Together with \cref{eq:carlen-lieb-1}, this completes the proof of the claimed
  inequality in \cref{eq:ell-strong-convex-2} which is equivalent the statement
  of the lemma.
\end{proof}

We mention that the inequality is not easy to establish without using the
Carlen-Lieb inequality.
For example, one may try to prove it using Golden-Thompson as it is already very
close in the form.
Yet, to complete the proof we would require that the exponential function is
operator convex, which is not the case unfortunately.
One may also try to use Jensen's trace inequality, but the strategy fails to
work because of the non-commutativity between the matrices.

\begin{theorem}\label{thm:parallel-kl}
  Let the sequence $\lambda^{(1)}, \lambda^{(2)}, \ldots$ be generated by
  \cref{alg:parallel-kl}.
  Then
  \begin{equation*}
    \lim_{t\to\infty} D \bigl( X_{0}, \L(Y_{0}, \lambda^{(t)} \cdot F)
    \bigr) = \inf_{\lambda \in \real^{k}} D \bigl( X_{0}, \L(Y_{0},
    \lambda \cdot F) \bigr).
  \end{equation*}
\end{theorem}

\begin{proof}
  It has been explained in the main text that the algorithm is derived from
  \cref{alg:parallel} by taking $f = x\ln(x) - x$.
  By \cref{lem:ell-strong-convex}, the assumption that $\ell_{f}$ is strong
  convex in this case is verified and the theorem follows from
  \cref{thm:parallel}.
\end{proof}

Similarly, \cref{thm:sequential} and \cref{lem:ell-strong-convex} guarantee the
following convergence theorem about the sequential version in
\cref{alg:sequential-kl}.

\begin{theorem}\label{thm:sequential-kl}
  Let the sequence $\lambda^{(1)}, \lambda^{(2)}, \ldots$ be generated by
  \cref{alg:sequential-kl}.
  Then
  \begin{equation*}
    \lim_{t\to\infty} D \bigl( X_{0}, \L(Y_{0}, \lambda^{(t)} \cdot F)
    \bigr) = \inf_{\lambda \in \real^{k}} D \bigl( X_{0}, \L(Y_{0},
    \lambda \cdot F) \bigr).
  \end{equation*}
\end{theorem}

\subsection{Quantum Partition Function Minimization and AdaBoost}

\begin{algorithm}[hbt!]
  \centering
  \begin{algspec}
    \textbf{Input:} $F=(F_{1}, F_{2}, \ldots, F_{k}) \in {\Herm(\X)}^{k}$
    and $\abs{F_{j}} \le \I$ for all $j$.\\
    \textbf{Output:} $\lambda^{(1)}, \lambda^{(2)}, \cdots$ such that
    \begin{equation*}
      \lim_{t\to\infty} \tr \bigl( \exp(\lambda^{(t)} \cdot F)
      \bigr) = \inf_{\lambda \in \real^{k}} \tr \bigl( \exp(\lambda \cdot F) \bigr).
    \end{equation*}
    \begin{algorithmic}[1]
      \State{Initialize $\lambda^{(1)} = (0, 0, \ldots, 0)$.} %
      \For{$t = 1, 2, \ldots, $} %
      \State{Compute $Y^{(t)} = \exp(\lambda^{(t)} \cdot F)$.} %
      \State{Compute $j_{t} = \argmax_{j} \abs{\ip{F_{j}}{Y^{(t)}}}$.%
        \label{al:argmax}} %
      \State{Compute $\delta^{(t)} = \dfrac{1}{2} \ln
        \dfrac{\ip{F_{j_{t}}^{-}}{Y^{(t)}}}{\ip{F_{j_{t}}^{+}}{Y^{(t)}}}$.} %
      \State{Update parameters
        $\lambda^{(t+1)}_{j} =
        \begin{cases}
          \lambda^{(t)}_{j} + \delta^{(t)} & \text{if } j = j_{t},\\
          \lambda^{(t)}_{j} & \text{otherwise.}
        \end{cases}$} %
      \EndFor{}
    \end{algorithmic}
  \end{algspec}
  \caption{Sequential Iterative Update Algorithm for Quantum Partition
    Function Minimization.}
  \label{alg:partition-function}
\end{algorithm}

In the following, we will discuss two important special cases of
\cref{alg:parallel-kl} (and its sequential analog \cref{alg:sequential-kl}).

The first case is given in \cref{alg:partition-function} presented in the
sequential update form.
It is an algorithm for minimizing the quantum partition function over a linear
family of Hamiltonians and is the non-commutative analog the famous AdaBoost
algorithm in learning theory.
The input to the algorithm is a tuple of Hermitian matrices
$F = (F_{1}, F_{2}, \ldots, F_{k})$, satisfying simple normalization conditions.
An important example is that the tuple consists of all terms of a local
Hamiltonian.
The algorithm iteratively computes a $\lambda$ so that the quantum partition
function $\tr(\exp(\lambda \cdot F))$ is minimized over the linear Hamiltonian family
$\sum_{j} \lambda_{j} F_{j}$ for $\lambda_{j} \in \real$.

To see that this is indeed a special case of \cref{alg:sequential-kl}, we choose
$X_{0} = 0$ and $Y_{0} = \I$.
In this case, the (non-normalized) Kullback-Leibler divergence
$D(X_{0}, \L(Y_{0}, \Lambda))$ simplifies to
\begin{equation*}
  D(0, \L(Y_{0}, \Lambda)) = \tr \L(Y_{0}, \Lambda) = \tr \exp(\Lambda),
\end{equation*}
and the equation in \cref{al:sequential-kl-eq} of \cref{alg:sequential-kl} has
solution
\begin{equation*}
  \delta^{(t)}_{j} = \frac{1}{2} \ln \frac{\ip{ F_{j}^{-} }{ Y^{(t)} }}{\ip{
      F_{j}^{+} }{ Y^{(t)} }}.
\end{equation*}
These calculations explains the changes made in \cref{alg:partition-function}
from \cref{alg:sequential-kl}.
The general convergence theorem \cref{thm:sequential} implies the following.

\begin{theorem}\label{thm:partition-function}
  Let $\lambda^{(1)}, \lambda^{(2)}, \ldots$ be the sequence generated by
  \cref{alg:partition-function}.
  Then
  \begin{equation*}
    \lim_{t\to\infty} \tr \bigl( \exp(\lambda^{(t)} \cdot F)
    \bigr) = \inf_{\lambda \in \real^{k}} \tr \bigl( \exp(\lambda \cdot F) \bigr).
  \end{equation*}
\end{theorem}

We now highlight the connection between the partition function minimization
algorithm in \cref{alg:partition-function} and the well-known AdaBoost algorithm
in learning theory.
Let $(x_{1}, y_{1}), (x_{2}, y_{2}), \ldots, (x_{m}, y_{m})$ be a set of
training examples where we assume for simplicity that each label $y_{i}$ is in
$\{\pm 1\}$.
There is a also a set of features or base hypothesis
$h_{1}, h_{2}, \ldots, h_{T}$ that predict the label $y_{i}$ when given $x_{i}$
as input.
The AdaBoost algorithm maintains a distribution over the training examples and
interacts with a weak learner.
In each step, the algorithm query the weak learner with the distribution and the
weak learner respond with a hypothesis $h_{j}$.
The algorithm then evaluates the performs of $h_{j}$ on each example and
increase the weight of the examples that are not correctly predicted by $h_{j}$
so that the weak learner is forced to focus more on those misclassified
examples.
Finally, the algorithm computes the parameters $\lambda \in \real^{n}$ based on
the error of each hypothesis $h_{j}$ and outputs a final hypothesis
\begin{equation*}
  H(x) = \mathrm{sign} \biggl(\sum_{j}^{T} \lambda_{j} h_{j}(x) \biggr).
\end{equation*}
We refer the readers to~\cite{FS97} for more details.

It is well known in the literature that the parameters $\lambda$ computed by the
AdaBoost algorithm actually minimizes the exponential loss defined as
\begin{equation}
  \label{eq:exp-loss}
  \sum_{i=1}^{m} \exp \biggl( -y_{i} \sum_{j=1}^{T} \lambda_{j} h_{j}
  (x_{i}) \biggr).
\end{equation}

Define diagonal matrices
\begin{equation*}
  F_{j} = \sum_{i=1}^{m} -y_{i} h_{j}(x_{i}) \ket{i}\bra{i},
\end{equation*}
for $j=1, 2, \ldots, T$.
The corresponding partition function is then
\begin{equation*}
  \tr \exp(\lambda \cdot F) = \tr \exp \biggl( \sum_{j=1}^{T}
  \lambda_{j} F_{j} \biggr),
\end{equation*}
which is the same as the exponential loss in \cref{eq:exp-loss}.
In this sense, our algorithm generalizes the iterative update procedure of
AdaBoost to the non-commutative setting.

\subsection{Quantum Iterative Scaling}

\begin{algorithm}[hbt!]
  \centering
  \begin{algspec}
    \textbf{Require:} $\rho_{0}, \sigma_{0} \in \Density(\X)$ such that
    $D(\rho_{0}, \sigma_{0}) < \infty$.\\
    \textbf{Input:} $F=(F_{1}, F_{2}, \ldots, F_{k}) \in {\Pos(\X)}^{k}$
    and $\sum_{j=1}^{k} F_{j} \le \I$.\\
    \textbf{Output:} $\lambda^{(1)}, \lambda^{(2)}, \cdots$ such that
    \begin{equation*}
      \lim_{t\to\infty} D \bigl( \rho_{0}, \L(\sigma_{0},
      \lambda^{(t)} \cdot F) \bigr) = \inf_{\lambda \in \real^{k}} D
      \bigl( \rho_{0}, \L(\sigma_{0}, \lambda \cdot F) \bigr).
    \end{equation*}
    \begin{algorithmic}[1]
      \State{Initialize $\lambda^{(1)} = (0, 0, \ldots, 0)$.} %
      \For{$t = 1, 2, \ldots, $} %
      \State{Compute $Y^{(t)} = \exp(\ln\sigma_{0} + \lambda^{(t)} \cdot F)$.}
      \For{$j = 1, 2, \ldots, k$} %
      \State{$\delta^{(t)}_{j} = \ln \ip{F_{j}}{\rho_{0}} -
        \ln \ip{F_{j}}{Y^{(t)}}$.}
      \EndFor{} %
      \State{Update parameters
        $\lambda^{(t+1)} = \lambda^{(t)} + \delta^{(t)}$.} %
      \EndFor{}
    \end{algorithmic}
  \end{algspec}
  \caption{Quantum iterative scaling algorithm.}
  \label{alg:qis}
\end{algorithm}
In the second special case, we consider the case where $F_{j}$ form a POVM and
$X_{0} = \rho_{0}$, $Y_{0} = \sigma_{0}$ are density matrices.
In that case, the solution of the equation in \cref{al:parallel-kl-eq} of
\cref{alg:parallel-kl} can be computed as
\begin{equation*}
  \delta^{(t)}_{j} = \ln \ip{F_{j}}{\rho_{0}} - \ln \ip{F_{j}}{Y^{(t)}}.
\end{equation*}
The \emph{quantum iterative scaling algorithm} given in \cref{alg:qis} then
follows naturally as a special case of \cref{alg:parallel-kl}.

The convergence proved in \cref{thm:parallel} then implies the following theorem
about the convergence of the quantum iterative scaling algorithm.
It is a non-commutative analog of the generalized iterative scaling algorithm
(also known as the SMART algorithm) as stated in Theorem 5.2 of~\cite{CS04}.
As in the commutative case, the intermediate matrices $Y^{(t)}$ are not
normalized to have trace one.
In fact, using the inequality in \cref{lem:ell-strong-convex}, it is easy to
show $\tr Y^{(t)} \le 1$ for all $t$.
Yet in the limit of $t \to \infty$, $Y^{(t)}$ converges to a density matrix.
\begin{theorem}\label{thm:qis}
  Let $\lf$ be the linear family defined by $\ip{F_{i}}{\rho} = \alpha_{i}$
  where $F_{i} \succeq 0$, $\sum_{i=1}^k F_{i} = \I$, and
  $\alpha_{i} = \ip{F_{i}}{\rho_{0}}$.
  Let $\rho_{0} \in \lf, \sigma_{0}$ be two density matrices such that
  $D(\rho_{0} \Vert \sigma_{0}) < +\infty$.
  Define a sequence of operators as
  \begin{equation*}
    Y^{(1)} = \sigma_{0}, \quad
    Y^{(t+1)} = \exp \Bigl( \ln Y^{(t)} + \sum_{i=1}^k (\ln \alpha_{i} -
    \ln \beta_{n,i}) F_{i} \Bigr),
  \end{equation*}
  where $\beta_{n,i} = \ip{F_{i}}{Y^{(t)}}$.
  Then the limit $\lim_{n\to \infty} Y^{(t)}$ converges to the information
  projection $\rho^{*}$ of $\sigma_{0}$ to linear family $\lf$.
\end{theorem}

An important special case is when $\sigma_{0} = \I/d$ where $d$ is the dimension
and $D(\rho, \sigma_{0}) = \ln(d) - S(\rho)$.
The Bregman projection of $\sigma_{0}$ to the linear family of $\rho_{0}$ and
$\{F_{j}\}$ is therefore the solution of the maximum entropy inference problem
formulated as
\begin{align*}
  \text{maximize:}\quad & S(\rho)\\
  \text{subject to:}\quad & \ip{F_j}{\rho} = \ip{F_{j}}{\rho_{0}},\\
                        & \rho \in \Density(\X).
\end{align*}
The sequence $Y^{(1)}, Y^{(2)}, \ldots$ computed in \cref{alg:qis} converges to
the solution of the above convex programming problem.
This shows that the quantum iterative scaling algorithm is an algorithm for
computing the maximum entropy inference given linear constraints of the density
matrix.
By Jaynes' maximum entropy principle, the maximum entropy state is the
exponential of a Hamiltonian of the form $\sum_{j} \lambda_{j} F_{j}$.
We call this problem of finding the Hamiltonian given local information of the
state as the Hamiltonian inference problem.
The general framework above proves the convergence of the algorithm for the
Hamiltonian inference problem and we leave the more detailed analysis of its
convergence rate as future work.

\section{Quantum Algorithmic Speedups}\label{sec:quantum}

In the previous discussions, the algorithms we presented are classical
algorithms that require matrix computations such as
$Y^{(t)} = \exp(\lambda^{(t)} \cdot F)$.
We will consider several ideas that can speedup the computation using techniques
from quantum algorithm design.

\subsection{Implement Matrix Functions on Quantum Computers}

The first natural attempt to quantize the algorithm is to implement the matrix
computations using techniques such as quantum singular value
transformation~\cite{GSLW19,Gil19} and smooth function evaluation~\cite{AGGW20}.

We start with the exact Bregman projection algorithm in \cref{alg:bregman}.
Notice that $Y^{(t)}$ in the algorithm is an intermediate quantity that are used
in later steps to approximate $\ip{F_{j}}{Y^{(t)}}$ and
$\ip{F_{j_{t}}}{\L_{f}(Y^{(t)}, \delta^{(t)} F_{j_{t}})}$.
It therefore suffices to have a subroutine $\abl$ that can compute the average
value $\ip{F_{j}}{\L_{f}(Y_{0}, \lambda \cdot F)}$ given $F_{j}$, $Y_{0}$, and
$\lambda$ as input or certain oracle access.
We emphasize that it is not necessary to compute the matrix $\L_{f}$ explicitly.
Numerical algorithms can then be employed to search for the solution of
$\delta^{(t)}$ in \cref{al:bregman-eq} of \cref{alg:bregman} as it is an
equation only involving a single real variable.

In the case of approximate Bregman projection algorithms, the situation is much
simpler as the equation involved is usually explicitly solvable and no numerical
search is necessary.
For example, in the case of approximate information projection algorithms
(\cref{alg:parallel-kl,alg:sequential-kl}), each iteration amounts to the
computation of $\ip{F_{j}}{Y^{(t)}}$ and $\ip{F_{j_{t}}^{\pm}}{Y^{(t)}}$ where
$Y^{(t)}$ has the exponential form $\exp(\log Y_{0} + \lambda^{(t)} \cdot F)$.
All these average values are in the form of the subroutine call to $\abl$.

Quantum algorithms that implement the subroutine $\abl$ are known.
As
\begin{equation*}
  \L_{f}(Y_{0}, \lambda \cdot F) = (f^{*})' (f'(Y_{0}) + \lambda \cdot F),
\end{equation*}
it is a matrix function of matrix $f'(Y_{0}) + \lambda \cdot F$.
Hence, under the condition that $f'(Y_{0}) + \lambda \cdot F$ is sparse, we can
apply the quantum algorithms for evaluating smooth functions of
Hamiltonians~\cite[Appendix B]{AGGW20}.
In the special case of Kullback-Leibler information projection algorithms, the
Bregman-Legendre projection is the exponential function $\exp(\lambda \cdot F)$
and their algorithm for approximating the average value
$\ip{F_{j}}{\exp(\lambda \cdot F)}$ runs in time
$\tilde{O}(\frac{\sqrt{n}Kd}{\theta})$ where $n$ is the size of the matrices,
$K$ is the upper bound of $\norm{\lambda \cdot F}$, $d$ is the sparsity of
$\lambda \cdot F$, and $\tilde{O}$ suppresses the polylogarithmic dependence on
the parameters.
The quantum running time is sometimes advantageous as classical algorithms for
computing the same quantity will run in time at least linear in $n$.
One caveat is that the quantum running time depends on the upper bound of the
Hamiltonian norm $K$.
In our case, $\lambda^{(t)}$ is being updated each iteration and the norm
$\norm{\lambda^{(t)} \cdot F}$ may even go to infinity for some problem
instances.
This prevents us to claim general time bounds using the techniques
in~\cite{AGGW20} but the quantum implementation could be advantageous in many
practical situations.
Other quantum algorithms for preparing the quantum Gibbs
states~\cite{PW09,CS16,MST+20} are also known with different assumptions and
performance guarantees and may be applicable as subroutines in our information
projection algorithms.

\subsection{Quantum Search for the Maximum Violation}

The second possible approach to speedup the iterative algorithms presented in
\cref{sec:algorithms} is to employ the fast quantum OR lemma as
in~\cite{HLM17,BKL+19} in sequential iterative algorithms.

Quantum implementations of $Y^{(t)}$ in the algorithm usually represent the
matrix as a quantum density matrix.
As quantum measurements may disturb the state they measure, a trivial approach to
estimate $\ip{F_{j}}{Y^{(t)}}$ requires a fresh copy of the state representing
$Y^{(t)}$ each time.
This could be very expensive as the preparation of $Y^{(t)}$ can be the hardest step
among the computations needed in each iteration.
The use of fast quantum OR lemma solves the problem by saving the number of
copies of states required.
Following the approach as in~\cite{BKL+19,AGGW20,AG19}, we show similar ideas
are applicable to our \emph{sequential} update algorithms.

To be more specific, we will consider the case of the partition function
minimization algorithm (\cref{alg:partition-function}) and focus on the cost of
\cref{al:argmax} in that algorithm.
We assume that there is a unitary $U_{j}$ that estimates the value of
$\ip{F_{j}}{\rho^{(t)}}$ to precision $\eta_{j}$ using $n^{(t)}$ copies of the state
$\rho^{(t)} = Y^{(t)}/ \tr(Y^{(t)})$.
Further assume that the access to $U_{j}$ is provided by a unitary $U$ such that
\begin{equation*}
  U \ket{j} \ket{\psi} = \ket{j} U_{j} \ket{\psi}.
\end{equation*}
A straightforward implementation of \cref{al:argmax} has to use fresh copies of
$\rho^{(t)}$ for different $j$ and the cost is
\begin{equation*}
  \Omega \Bigl(m T(U) + m n^{(t)} T_{\rm State} (\rho^{(t)}) \Bigr),
\end{equation*}
where $T_{\rm State}(\rho^{(t)})$ is the time compelxity of preparing the Gibbs
state $\rho^{(t)}$ and $m=k$ is the number of constraints.
Applying the \emph{two-phase quantum minimum finding} algorithm from~\cite[Lemma
7]{AG19}, the time complexity is improved to
\begin{equation*}
  \tilde{O} \Bigl(\sqrt{m} T(U) +
  \log^{4}(m)\log(1/\delta) n^{(t)} T_{\rm State}(\rho^{(t)}) \Bigr)
\end{equation*}
where $\delta$ is the precision parameter.

\subsection{NISQ Applications}

We now briefly mention potential applications of the above algorithms in the
context of designing algorithms on noisy intermediate scale quantum
devices~\cite{Pre18}.

Assume that there are NISQ algorithms for approximately preparing a quantum
state representing the matrix $Y^{(t)}$ in the algorithm.
For the preparation of quantum Gibbs states, it is shown in~\cite{CLW20,MS19}
that a variational quantum algorithm for preparing the Gibbs state can be
derived using a gradient descent method optimizing the free energy.
Then \cref{alg:parallel-kl,alg:sequential-kl} can also be efficiently
implemented on a NISQ device as our iterative algorithm has the structure of
hybrid quantum algorithms with a quantum part for the preparation of the quantum
state for $Y^{(t)}$ and classical part that updates the parameters using a
simple rule based on the average values of the current state.
In some sense, it is a variational quantum algorithm in which the quantum part
is a Gibbs state preparation subroutine and the classical update rules are given
by our algorithmic framework.
In the case of quantum iterative scaling, this approach lifts the entropy
estimation algorithm in~\cite{CLW20} to a solver for maximum entropy problem
with linear constraints.

\section{Summary}\label{sec:summary}

In this paper, we prove a general duality theorem for Bregman divergence on
Hermitian matrices under simple assumptions of the underlying convex function.
Several iterative update algorithms are designed based on the idea of exact and
approximate Bregman projections and the convergence is proved using the duality
theorem and the auxiliary function method.

There are many interesting questions left open and we leave them as future work.
First, we have only been able to prove the strong convex inequality for
Kullback-Leibler divergence.
Can we have a more general theory about the condition under which the strong
convex is true?
Better understanding of related inequalities will lead to, for example, the
non-commutative analog of logistic regression based on Fermi/Dirac convex
function (Line 4 of \cref{tab:examples})~\cite{BB97,CSS02}.
Second, we have only been able to work with linear equality constraints and it
is an interesting problem to further generalize the framework so that we can
handle linear inequality constraints.
Finally, it is an interesting problem to establish quantitative bounds on the
class and quantum time complexity of the algorithms introduced in the paper.

\section*{Acknowledgments}

The author thanks Bei Zeng for stimulating discussions on related problems and
for bringing~\cite{CS04} to his attention.
He thanks Tongyang Li for helpful discussions and for pointing out the related
work of~\cite{HKT21}.

\bibliographystyle{abbrv-links}

\bibliography{quantum-iterative-scaling}


\end{document}